\documentclass[letterpaper,journal,draftclsnofoot,onecolumn]{IEEEtranBB}   

\usepackage[T1]{fontenc}
\usepackage{graphicx}
	\graphicspath{{figs/}{matlab/}}   
\usepackage[cmex10]{amsmath}
\usepackage{amssymb, amsthm, gensymb}  
\usepackage{booktabs}  
	\renewcommand{\tabcolsep}{4mm}  
\usepackage{cellspace}
\usepackage{float}
\usepackage{flafter}  
\usepackage{accents}
\usepackage{cases}  
\usepackage[usenames,dvipsnames]{color}

\usepackage{calc}

\newcommand{\fn}{\footnotesize}  

\newcommand{\executeiffilenewer}[3]{%
\ifnum\pdfstrcmp{\pdffilemoddate{#1}}%
{\pdffilemoddate{#2}}>0%
{\immediate\write18{#3}}\fi%
}
\newcommand{\includesvg}[1]{%
\executeiffilenewer{#1.svg}{#1.pdf}%
{inkscape -z -D --file=#1.svg %
--export-pdf=#1.pdf --export-latex}%
\ifx\svgscale\undefined%
	\newcommand{\svgscale}{\defaultsvgscale}
\fi%
\input{#1.pdf_tex}%
}

	\newcommand{\defaultsvgscale}{1.0}  

\usepackage[tight,small,bf]{subfigure}   

\usepackage{microtype}

\usepackage{xspace}
\usepackage{bbm}
\usepackage{mathrsfs}
\input{mathlig}

\newcommand{\muspace}{\mspace{1mu}}

\DeclareRobustCommand{\scond}{\mathchoice{\muspace\vert\muspace}{\vert}{\vert}{\vert}}
\mathlig{|}{\scond}
\newcommand{\cond}{\mathchoice{\,\vert\,}{\mspace{2mu}\vert\mspace{2mu}}{\vert}{\vert}}

\DeclareRobustCommand{\discint}{\mathchoice{\mspace{-1.5mu}:\mspace{-1.5mu}}{\mspace{-1.5mu}:\mspace{-1.5mu}}{:}{:}}
\mathlig{::}{\discint}
\newcommand{\suchthat}{\mathchoice{\colon}{\colon}{:\mspace{1mu}}{:}}

\newcommand{\Cc}{\mathcal{C}}

\newcommand{\Dc}{\mathcal{D}}
\newcommand{\Ec}{\mathcal{E}}

\newcommand{\Lc}{\mathcal{L}}

\newcommand{\Qc}{\mathcal{Q}}

\newcommand{\Sc}{\mathcal{S}}
\newcommand{\Tc}{\mathcal{T}}
\newcommand{\Uc}{\mathcal{U}}
\newcommand{\Vc}{\mathcal{V}}

\newcommand{\Xc}{\mathcal{X}}
\newcommand{\Yc}{\mathcal{Y}}

\newcommand{\Cr}{\mathscr{C}}

\newcommand{\Rr}{\mathscr{R}}

\newcommand{\pen}{{P_e^{(n)}}}

\newcommand{\aep}{{\mathcal{T}_{\epsilon}^{(n)}}}

\newcommand{\Mh}{{\hat{M}}}

\newcommand{\mh}{{\hat{m}}}

\def\e{\epsilon}
\def\eps{\epsilon}

\DeclareMathOperator\E{\textsf{E}}
\let\P\relax
\DeclareMathOperator\P{\textsf{P}}

\def\error{\mathrm{e}}

\newcommand{\N}{\mathrm{N}}
\newcommand{\U}{\mathrm{Unif}}

\def\textiid{i.i.d.\@\xspace}
\newcommand\iid{\ifmmode\text{ i.i.d. } \else \textiid \fi}

\newcommand{\sfrac}[2]{\mbox{\small$\displaystyle\frac{#1}{#2}$}}

\def\mathllap{\mathpalette\mathllapinternal}
\def\mathllapinternal#1#2{%
  \llap{$\mathsurround=0pt#1{#2}$}}

\def\clap#1{\hbox to 0pt{\hss#1\hss}}
\def\mathclap{\mathpalette\mathclapinternal}
\def\mathclapinternal#1#2{%
  \clap{$\mathsurround=0pt#1{#2}$}}

\let\oldstackrel\stackrel
\renewcommand{\stackrel}[2]{\oldstackrel{\mathclap{#1}}{#2}}

\renewcommand{\hbar}{h\mathllap{\overline{\vphantom{h}\hphantom{\rule{4.6pt}{0pt}}}\mspace{0.77mu}}}

\catcode`~=11 
\newcommand{\urltilde}{\kern -.06em\lower -.06em\hbox{~}\kern .02em}
\catcode`~=13 

\usepackage{caption}
\DeclareCaptionLabelSeparator{periodquad}{. \quad}
\ifCLASSOPTIONonecolumn
	\ifCLASSOPTIONtrimmargin
		\usepackage{anysize}
		\papersize{9.7in}{5.9in}  
		\marginsize{0.2in}{0.2in}{0.1in}{0.1in} 
	\else
	\fi
\fi

\usepackage[colorlinks=true,linkcolor=black,anchorcolor=black,citecolor=black,filecolor=black,menucolor=black,urlcolor=black]{hyperref}  
	
\usepackage{cite}  

\renewcommand{\epsilon}{\varepsilon} 
\renewcommand{\eps}{\varepsilon}

\providecommand{\card}[1]{\lvert#1\rvert}

\definecolor{unemphColor}{gray}{0.5}  

\newcommand{\anneq}[1]{\overset{\text{(#1)}}{=}}  
\newcommand{\annleq}[1]{\overset{\text{(#1)}}{\leq}}  

\newcommand{\natSet}[1]{[1::#1]}  
\newcommand{\natSetFromZero}[1]{[0::#1]}  

\renewcommand{\U}{\mathrm{Unif}}

\newcommand{\Er}{\mathcal E}

\newcommand{\compl}[1]{#1^{\mathrm{c}}}

\theoremstyle{definition}  
\newtheorem{thm}{Theorem}
\newtheorem{definition}{Definition}    
\newtheorem{lemma}{Lemma}
\newtheorem{corollary}{Corollary}
\newtheorem{proposition}{Proposition}

\theoremstyle{remark}
\newtheorem{remark}{Remark}
\newtheorem{example}{Example}

\newcommand{\ReduceA}{\mathop{\mathrm{Proj}_{4\to 2}}}

\newcommand{\Scs}{{\Sc^\star}}
\newcommand{\Scn}{{\overline{\Sc}^\star}}

\newcommand{\inpmfspace}{\hspace{0.35ex}}

\title{Optimal Achievable Rates for \\ Interference Networks with Random Codes}
\author{%
Bernd Bandemer, Abbas El Gamal, and Young-Han Kim
\thanks{\hrule \vspace{2mm} \noindent This research was supported in part by the Korea Communications Commission under the R\&D program KCA-2012-11-921-04-001 (ETRI).

\noindent This paper was presented in part at the Allerton Conference on Communication, Control and Computing, October 2012, Monticello, IL.
}
}
\hypersetup{
	pdfauthor = {Bernd Bandemer, Abbas El Gamal, and Young-Han Kim},
	pdftitle = {Optimal Achievable Rates for Interference Networks with Random Codes},
	pdfsubject = {Information theory},
	pdfkeywords = {},
	pdfcreator = {LaTeX}
}

\begin{document}
\maketitle


\begin{abstract}
The optimal rate region for interference networks is characterized when encoding is restricted to random code ensembles with superposition coding and time sharing.
A simple \emph{simultaneous nonunique decoding} rule, under which each receiver decodes for the intended message as well as the interfering messages, is shown to achieve this optimal rate region regardless of the relative strengths of signal, interference, and noise. 
This result implies that the Han--Kobayashi bound, the best known inner bound on the capacity region of the two-user-pair interference channel, cannot be improved merely by using the optimal maximum likelihood decoder.
\end{abstract}


\section{Introduction}
Consider a communication scenario in which multiple senders communicate independent messages 
to multiple receivers over a network with interference. What is the set of simultaneously achievable rate tuples for reliable communication? What coding scheme achieves this \emph{capacity region}? Answering these questions involves joint optimization of the encoding and decoding functions, which has remained elusive even for the case of two sender--receiver pairs.

With a complete theory in terra incognita, in this paper we take a simpler modular approach to these questions. Instead of searching for the optimal encoding functions,
suppose rather that the encoding functions are restricted to realizations of a given random code ensemble of a certain structure. What is the set of simultaneously achievable rate tuples so that the probability of decoding error,
when \emph{averaged over the random code ensemble}, can be made arbitrarily small?
To be specific, we focus on random code ensembles with superposition coding 
and time sharing of independent and identically distributed (i.i.d.) codewords.
This class of random code ensembles
includes those used in the celebrated Han--Kobayashi coding
scheme~\cite{HanKobayashi81}.

We characterize the set $\Rr^*$ of rate tuples achievable by the random code ensemble for an interference network
as the intersection of rate regions for its component multiple access channels in which each receiver recovers its intended
messages as well as appropriately chosen unintended messages.
More specifically, the rate region $\Rr^*$ for the
interference network with senders $\natSet K = \{1,2,\ldots,K\}$, each communicating an independent message, and receivers $\natSet L$, each required to recover a subset $\Dc_1, \ldots, \Dc_L \subseteq \natSet K$ of messages,
is
\begin{equation}
\Rr^* = \bigcap_{l \in \natSet L} \,
\bigcup_{\Sc \supseteq \Dc_l} \Rr_{\text{MAC}(\Sc,\,l)}.
\end{equation}
Here $\Rr_{\text{MAC}(\Sc,\,l)}$ denotes the set of rate tuples achievable by the random code ensemble for the \emph{multiple
access channel} with senders $\Sc$ and receiver $l$ when the codewords from the other senders
$\natSet K \setminus \Sc$ are treated as random noise.

A direct approach to proving this result would be to analyze
the average performance of the optimal decoding rule for each realization of the random code ensemble
that minimizes the probability of decoding error, namely, maximum likelihood decoding (MLD).
This analysis, however, is unnecessarily cumbersome.
We instead take an indirect yet conducive approach that is common in information theory.
First, we show that any rate tuple inside $\Rr^*$ is achieved by
using the typicality-based \emph{simultaneous nonunique decoding} (SND) rule~\cite{ChongMotani08,Nair2009,ElGamalKim},
in which each receiver attempts to recover the codewords from its intended senders
and (potentially nonuniquely) the codewords from interfering senders.
Second, we show that if the average probability of error of MLD for the random code ensemble is asymptotically zero,
then its rate tuple must lie in $\Rr^*$.
The key to proving the second step is to show that after a maximal set of messages has been recovered, the remaining signal 
at each receiver is distributed essentially independently and identically.
The two-step approach taken here is reminiscent of
the random coding proof for the capacity of the point-to-point channel~\cite{Shannon1948}, wherein a suboptimal (in the sense of the probability of error) decoding rule based on the notion of joint typicality can achieve the same
rate as MLD when used for random code ensembles.

Our result has several implications.
\begin{itemize}
\item It shows that 
incorporating the structure of interference into decoding,
when properly done as in MLD and SND,
always achieves higher or equal rates
compared to treating interference as random noise; thus, 
the traditional wisdom of distinguishing
between decoding for the interference at high signal-to-noise ratio and ignoring the interference 
at low signal-to-noise ratio does not provide any improvement on achievable rates.

\item It shows that the Han--Kobayashi inner bound~\cite{HanKobayashi81}, \cite{ChongMotani08}, \cite[Theorem 6.4]{ElGamalKim}, which was established using the random code ensemble and a typicality-based 
simultaneous decoding rule, cannot be improved by using a more powerful decoding rule such as MLD.

\item It generalizes the result by Motahari and Khandani~\cite{Motahari2011}, and  Baccelli, El Gamal, and Tse~\cite{Baccelli2011}
on the optimal rate region for $K$-user-pair Gaussian interference channels with point-to-point Gaussian 
random code ensembles to arbitrary (not necessarily Gaussian) random code ensembles
with time sharing and superposition coding.

\item It shows that the Cover--van der Meulen inner bound with no common auxiliary random variable on the capacity region of the two-receiver broadcast channel~\cite{Cover1975}, \cite{vanderMeulen1975},  \cite[Eq.~(8.8)]{ElGamalKim} (and thus Marton's inner bound~\cite{Marton1979}, \cite[Theorem~8.3]{ElGamalKim}) can be improved 
by using SND to include the superposition coding inner bound~\cite{Cover1972}, \cite{Bergmans1973}, 
\cite[Theorem 5.1]{ElGamalKim}.

\item It shows that the interference decoding rate region for the three-user-pair deterministic interference channel in~\cite{Bandemer:2010:InterferenceDecodingJournal} is the optimal rate region achievable by point-to-point random code ensembles and time sharing.
\end{itemize}

We illustrate the main result and its implications via the following two simple examples.

\subsection{Interference Channels with Two User Pairs}  \label{sec:intro_2dmic}
Consider the two-user-pair discrete memoryless interference channel (2-DM-IC) $p(y_1,y_2|x_1,x_2)$ with input alphabets $\Xc_1$ and $\Xc_2$ and output alphabets $\Yc_1$ and $\Yc_2$, depicted in Figure~\ref{fig:2dmic}.
Here sender~$j = 1,2$ wishes to communicate a message to its respective receiver 
via $n$ transmissions over the shared interference channel. Each message $M_j$, $j = 1,2$, is separately encoded into
a codeword $X_j^n = (X_{j1}, X_{j2}, \ldots, X_{jn})$ and transmitted over the channel. Upon receiving the
sequence $Y_j^n$, receiver $j=1,2$ finds an estimate $\Mh_j$ of the message $M_j$.
\begin{figure}[h]
	\centering
	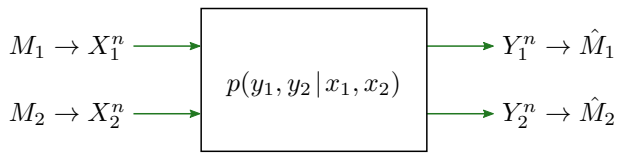
	\caption{Two-user-pair discrete memoryless interference channel.}
	\label{fig:2dmic}
\end{figure}

We now consider the standard random coding analysis for inner bounds 
on the set of achievable rate pairs (the capacity region) of the 2-DM-IC.
Given a product input pmf $p(x_1)\inpmfspace p(x_2)$, suppose that 
the codewords $x_j^n(m_j)$, $m_j \in \natSet{2^{nR_j}} = \{1,2,\ldots,2^{nR_j}\}$, for $j=1,2$ are generated randomly, 
each drawn according to $\prod_{i=1}^n p_{X_j}(x_{ji})$.

We recall the rate regions achieved by employing the following simple suboptimal decoding rules, 
described for receiver~$1$ (cf.~\cite[Sec.~6.2]{ElGamalKim}).
\begin{itemize}
\item \emph{Treating interference as noise (IAN).} Receiver 1 finds the unique message
$\mh_1$ such that $(x_1^n(\mh_1), y_1^n)$ is jointly typical. (See the end of this section
for the definition of joint typicality.)
It can be shown that the average probability of decoding error for receiver~$1$
tends to zero as $n \to \infty$ if
\begin{equation} 
R_1 < I(X_1; Y_1).    \label{eq:ian-bound}
\end{equation}
The corresponding rate region (IAN region) is depicted in Figure~\ref{fig:2dmic_regionR1_IAN}.

\item \emph{Simultaneous decoding (SD).} Receiver 1 finds the unique message pair
$(\mh_1,\mh_2)$ such that $(x_1^n(\mh_1), x_2^n(\mh_2), y_1^n)$ is jointly typical. The average probability of decoding error for receiver~$1$
tends to zero as $n \to \infty$ if
\begin{subequations}
\begin{align}
R_1 &< I(X_1; Y_1 \cond X_2),\\
R_2 &< I(X_2; Y_1 \cond X_1),\\
R_1 + R_2 &< I(X_1, X_2; Y_1).
\end{align}
\label{eq:sd-bound}%
\end{subequations}
The corresponding rate region (SD region) is depicted in Figure~\ref{fig:2dmic_regionR1_SD}.
\end{itemize}

Now, consider simultaneous nonunique decoding (SND) in which receiver~1 finds the unique
$\mh_1$ such that $(x_1^n(\mh_1), x_2^n(m_2), y_1^n)$ is jointly typical \emph{for some $m_2$}.
Clearly, any rate pair in the SD rate region~\eqref{eq:sd-bound} is achievable via SND.
Less obviously, any rate pair in the IAN region~\eqref{eq:ian-bound} is also achievable via SND
as we show in the achievability proof of Theorem~\ref{thm:2dmic} in Section~\ref{sec:2dmic}.
Hence, SND can achieve any rate pair in
the union of the IAN and SD regions, that is, the rate region $\Rr_1$ as depicted in
Figure~\ref{fig:2dmic_regionR1_simple}. Similarly, the average probability of decoding error for receiver~$2$ using SND tends to zero as $n \to \infty$ if the rate pair $(R_1,R_2)$ is in $\Rr_2$, which is defined analogously by exchanging the roles of the two users (see Figure~\ref{fig:2dmic_regionR2_simple}).
Combining the decoding requirements for both receivers yields the rate region $\Rr_1 \cap \Rr_2$.

This rate region $\Rr_1 \cap \Rr_2$ turns out to be optimal for the given random code ensemble. As
shown in the converse proof of Theorem~\ref{thm:2dmic}, if the probability of error for MLD averaged
over the random code ensemble tends to zero as $n \to \infty$, then the rate pair $(R_1,R_2)$ must reside
inside the closure of $\Rr_1 \cap \Rr_2$. Thus, SND achieves the same rate region as MLD (for random code ensembles of the given structure).

\begin{figure}[h!]
	\vspace*{5mm}
	\begin{tabular}{cc}
	\subfigure[]{
		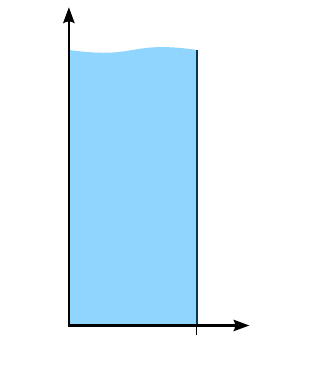
		\label{fig:2dmic_regionR1_IAN}
	}
	&
	\subfigure[]{
		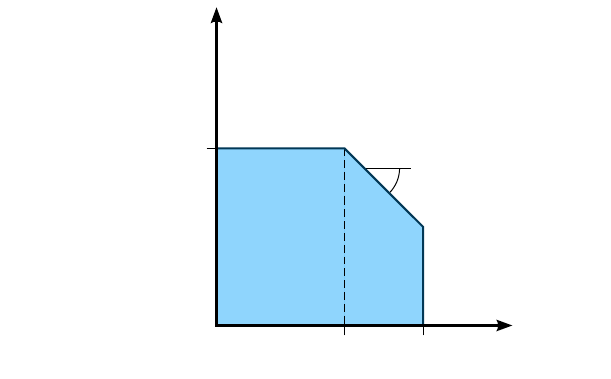
		\label{fig:2dmic_regionR1_SD}
	} 
	\\[5mm]
	\subfigure[]{
		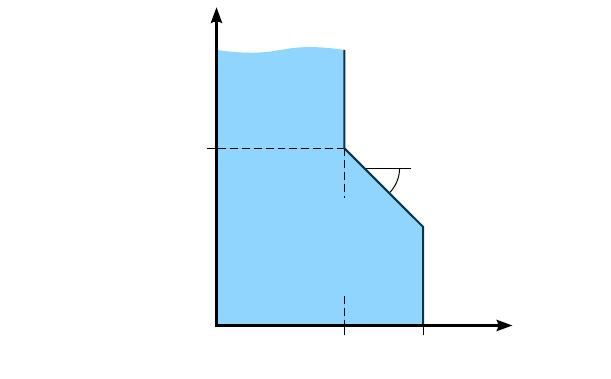
		\label{fig:2dmic_regionR1_simple}
	}
	&
	\subfigure[]{
		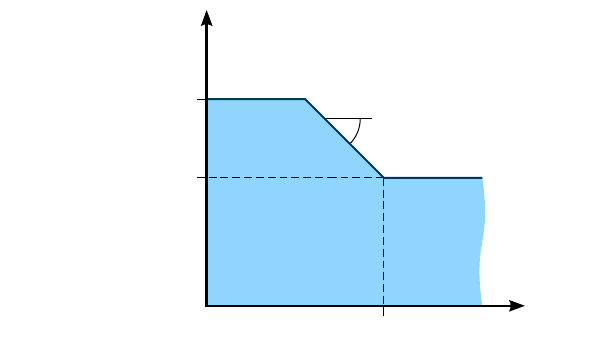
		\label{fig:2dmic_regionR2_simple}
	}
	\end{tabular}
	\caption{Achievable rate regions for the 2-DM-IC: (a)~treating interference as noise, (b)~using simultaneous decoding, (c)~using simultaneous nonunique decoding ($\Rr_1$); note that $\Rr_1$ is the union of the regions in (a) and (b); and (d)~using simultaneous nonunique decoding at receiver 2 ($\Rr_2$).}
	\label{fig:2dmic_region}
\end{figure}

\subsection{Broadcast Channels with Two Receivers} \label{sec:intro_2bc}
In the previous example, the random code ensemble for each sender 
had the structure of random code ensembles for point-to-point communication channels~\cite{Shannon1948}.
To illustrate our result for superposition coding, consider 
the two-receiver discrete memoryless broadcast channel (2-DM-BC) $p(y_1,y_2|x)$ with input alphabet $\Xc$ and output alphabets $\Yc_1$ and $\Yc_2$. 
Here the sender wishes to communicate two independent messages to their respective receivers 
via $n$ transmissions over the broadcast channel. Each message pair $(M_1, M_2)$ is encoded into
a codeword $X^n$ and transmitted over the channel. Upon receiving the
sequence $Y_j^n$, receiver $j=1,2$ finds an estimate $\Mh_j$ of the message $M_j$.

We consider a special case of the classical coding scheme by Cover~\cite{Cover1975} and van der Meulen~\cite{vanderMeulen1975}, illustrated in
Figure~\ref{fig:HK22dmic}. Given
a product pmf $p(u_1)\inpmfspace p(u_2)$ and a function $x(u_1,u_2)$,
suppose that the codewords $x^n(m_1,m_2)$, $(m_1,m_2) \in \natSet{2^{nR_1}} \times \natSet{2^{nR_2}}$,
are given as $x_i(m_1,m_2) = x(u_{1i}(m_1), u_{2i}(m_2))$, $i \in \natSet n$,
where the sequences $u_j^n(m_j)$, $m_j \in \natSet{2^{nR_j}}$, for $j=1,2$ are generated randomly, 
each drawn according to $\prod_{i=1}^n p_{U_j}(u_{ji})$. Thus, the transmitted codeword is a ``superposition'' 
of two codewords $u_1^n(m_1)$ and $u_2^n(m_2)$, which is literally the case when $x(u_1,u_2)$ is additive.

\begin{figure}[h]
	\centering
	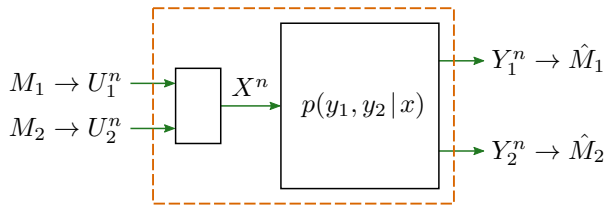
	\caption{Broadcast channel with Cover--van der Meulen coding.}
	\label{fig:HK22dmic}
\end{figure}

Alternatively, this superposition coding scheme can be viewed as first transforming the underlying
the broadcast channel into a two-user-pair interference channel 
\[
p(y_1,y_2 | u_1,u_2) = p(y_1, y_2 | x(u_1, u_2) )
\]
and then applying the random coding scheme for two-user-pair interference channel discussed in Subsection~\ref{sec:intro_2dmic}.
Hence, the random coding analysis thereof can be readily applied. For example, suppose that each receiver
decodes for its intended codeword while treating the other codeword as noise (cf.~\eqref{eq:ian-bound}).
Then it can be shown that the average probability
of decoding error tends to zero as $n \to \infty$ if 
\begin{subequations}
\label{eq:2BC_cvdm}
\begin{align}
R_1 &< I(U_1; Y_1),\\
R_2 &< I(U_2; Y_2).
\end{align}
\end{subequations}
Taking the union over all pmfs
$p(u_1)\inpmfspace p(u_2)$ and functions $x(u_1,u_2)$, we obtain the Cover--van der Meulen inner bound
(with no
common auxiliary random variable) on the capacity region~\cite[Eq.~(8.8)]{ElGamalKim}.

On the other hand, consider the superposition coding inner bound on the capacity region~\cite{Cover1972}, \cite{Bergmans1973},
\cite[Theorem~5.1]{ElGamalKim}, which is the set of rate pairs such that
\begin{subequations}
\label{eq:2BC_superpos}
\begin{align}
R_1 &< I(U_1; Y_1 | U_2) = I(X; Y_1 | U_2),\\
R_2 &< I(U_2; Y_2),\\
R_1 + R_2 &<  I(U_1,U_2; Y_1) = I(X; Y_1)
\end{align}
\end{subequations}
for some pmf $p(u_1)\inpmfspace p(u_2)$ and function $x(u_1,u_2)$. This inner bound 
corresponds to having receiver~1 decode for both messages while receiver~2 treats the other codeword as noise. It can be shown~\cite{Gohari2010} that this bound is not in general contained in the Cover--van der Meulen inner bound and neither vice versa. (This statement remains true even if the Cover--van der Meulen inner bound is replaced with Marton's inner bound without a common auxiliary random variable~\cite{Marton1979}, \cite[Theorem~8.3]{ElGamalKim}).

The distinction between the superposition coding inner bound and the Cover--van der Meulen inner bound 
is, however, a mere side effect from the use of suboptimal decoding rules.
Suppose now that both receivers use SND. As in Subsection~\ref{sec:intro_2dmic},
the average probability of decoding error tends to
zero as $n \to \infty$ if $(R_1,R_2) \in \Rr_1 \cap \Rr_2$, where $\Rr_1$ consists of rate pairs such that
\begin{align*} 
R_1 &< I(U_1; Y_1)\\
\intertext{or} 
R_1 &< I(U_1; Y_1 \cond U_2),\\
R_1 + R_2 &< I(U_1, U_2; Y_1),
\end{align*}
and $\Rr_2$ is similarly defined by exchanging the subscripts $1$ and $2$. 
The union of $\Rr_1 \cap \Rr_2$ over all pmfs $p(u_1)\inpmfspace p(u_2)$ and functions $x(u_1,u_2)$ yields an inner bound on the capacity region. It is not hard to see that this region includes both inner bounds~\eqref{eq:2BC_cvdm} and~\eqref{eq:2BC_superpos}. Furthermore, this region is the optimal rate region 
achieved by using MLD (see Section~\ref{sec:KLdmic}). 

\vspace{2mm} 

The rest of the paper is organized as follows. For simplicity of presentation, in Section~\ref{sec:2dmic} we formally define the problem for the two-user-pair interference channel
and establish our main result for the random code ensemble with time sharing and no superposition coding. 
In Section~\ref{sec:KLdmic}, we extend our result to a multiple-sender multiple-receiver discrete memoryless interference network in which each sender has a single message and wishes to communicate it to a subset of the receivers. This extension includes superposition coding with an arbitrary number of layers. In Section~\ref{sec:HK}, we specialize the result to the Han--Kobayashi coding scheme for the two-user-pair interference channel. Most technical proofs are deferred to the Appendices.

Throughout we closely follow the notation in \cite{ElGamalKim}. 
In particular, for $X \sim p(x)$ and $\e \in (0,1)$, we define the set of $\e$-typical $n$-sequences
$x^n$ (or the typical set in short)~\cite{Orlitsky2001} as
\[
\aep(X) = \bigl\{ x^n : | \#\{i \suchthat x_i = x\}/n - p(x) | \le \e p(x)
\text{ for all } x \in \Xc \bigr\}.
\]
For a tuple of random variables $(X_1,\ldots,X_k)$,
the joint typical set $\aep(X_1,\ldots,X_k)$ is defined as the typical set
$\aep((X_1,\ldots,X_k))$
for a single random variable 
$(X_1,\ldots,X_k)$. The joint typical set
$\aep(X_\Sc)$ for a subtuple $X_\Sc = (X_k: k \in \Sc)$
is defined similarly for each $\Sc \subseteq \natSet k$. We use $\delta(\e) > 0$ to denote a generic function
of $\e > 0$ that tends to zero as $\e \to 0$. Similarly, we use $\e_n \ge 0$ to denote a generic function
of $n$ that tends to zero as $n \to \infty$.

\section{DM-IC with Two User Pairs}  \label{sec:2dmic}

Consider the two-user-pair discrete memoryless interference channel 
(2-DM-IC) $p(y_1,y_2\cond x_1,x_2)$ introduced in Subsection~\ref{sec:intro_2dmic} (see Figure~\ref{fig:2dmic}). 
A $(2^{nR_1}, 2^{nR_2}, n)$ code $\Cc_n$ for the 2-DM-IC consists of
\begin{itemize}
\item
two message sets $\natSet{2^{nR_1}}$ and $\natSet{2^{nR_2}}$,

\item two encoders, where encoder~1 assigns a codeword $x_1^n(m_1)$ to
each message $m_1 \in \natSet{2^{nR_1}}$ and encoder~2 assigns a codeword
$x_2^n(m_2)$ to each message $m_2\in \natSet{2^{nR_2}}$, and

\item
two decoders, where decoder~1 assigns an estimate $\mh_1$ or an
error message $\error$ to each received sequence $y_1^n$ and decoder~2
assigns an estimate $\mh_2$ or an error message $\error$ to each
received sequence $y_2^n$.
\end{itemize}
We assume that the message pair $(M_1,M_2)$ is uniformly distributed over
$\natSet{2^{nR_1}} \times \natSet{2^{nR_2}}$.  The average probability of
error for the code $\Cc_n$ is defined as
\[
\pen(\Cc_n) = \P\bigl\{ (\Mh_1, \Mh_2) \ne (M_1,M_2) \bigr\}.
\]
A rate pair $(R_1,R_2)$ is said to be \emph{achievable} for the 2-DM-IC
if there exists a sequence of $(2^{nR_1}, 2^{nR_2}, n)$ codes $\Cc_n$ such that
$\lim_{n \to \infty} \pen(\Cc_n) = 0$. The \emph{capacity region} 
$\Cr$ of the 2-DM-IC is the closure of the set of achievable rate pairs $(R_1, R_2)$.

We now limit our attention to a randomly generated code ensemble with a special structure.
Let $p = p(q,x_1,x_2) = p(q)\inpmfspace p(x_1|q)\inpmfspace p(x_2|q)$ be a given pmf on $\Qc \times \Xc_1 \times \Xc_2$, where
$\Qc$ is a finite alphabet. 
Suppose that the codewords 
$X_1^n(m_1)$, $m_1 \in \natSet{2^{nR_1}}$, and
$X_2^n(m_2)$, $m_2 \in \natSet{2^{nR_2}}$, that constitute the codebook, are
generated randomly as follows:
\begin{itemize}
\item Let $Q^n \sim \prod_{i=1}^n p_Q(q_i)$. 
\item Let $X_1^n(m_1) \sim \prod_{i=1}^n p_{X_1|Q}(x_{1i}|q_i)$, $m_1 \in \natSet{2^{nR_1}}$,
conditionally independent given $Q^n$.
\item Let $X_2^n(m_2) \sim \prod_{i=1}^n p_{X_2|Q}(x_{2i}|q_i)$, $m_2 \in \natSet{2^{nR_2}}$, 
conditionally independent given $Q^n$.
\end{itemize}
Each instance $\{(x_1^n(m_1),x_2^n(m_2)): (m_1,m_2) \in \natSet{2^{nR_1}} \times \natSet{2^{nR_2}}\}$
of such generated codebooks, along with the corresponding optimal decoders, constitutes
a $(2^{nR_1}, 2^{nR_2}, n)$ code. We refer to the random code ensemble generated in this
manner as the $(2^{nR_1}, 2^{nR_2}, n; p)$ \emph{random code ensemble}.

\begin{definition}[Random coding optimal rate region]
Given a pmf $p = p(q)\inpmfspace p(x_1|q)\inpmfspace p(x_2|q)$,
the \emph{optimal rate region $\Rr^*(p)$ achievable by the $p$-distributed random
code ensemble} is the closure of the set of rate pairs $(R_1,R_2)$ such that
the sequence of $(2^{nR_1},2^{nR_2},n; p)$ random code ensembles $\Cc_n$ satisfies
\[
\lim_{n\to\infty} \E_{\Cc_n}[ \pen(\Cc_n) ]= 0,
\] 
where the expectation is with respect to the random code ensemble $\Cc_n$.
\end{definition}

To characterize the random coding optimal rate region, 
we define $\Rr_1(p)$ to be the set of rate pairs $(R_1,R_2)$ such that
\begin{subequations} \label{eq:2dmic-r1}
\begin{align}
		R_1 &\le I(X_1;Y_1 \cond Q) \\
\intertext{or}
		R_2 &\le I(X_2;Y_1 \cond X_1,Q), \\
		R_1+R_2 &\le I(X_1,X_2; Y_1 \cond Q).
\end{align}
\end{subequations}
Similarly, define $\Rr_2(p)$ by making the index substitution $1 \leftrightarrow 2$.
%
We are now ready to state the main result of the section.

\begin{thm} \label{thm:2dmic}
Given a pmf $p = p(q)\inpmfspace p(x_1|q)\inpmfspace p(x_2|q)$,
the optimal rate region of the DM-IC $p(y_1,y_2|x_1,x_2)$
achievable by the $p$-distributed random code ensemble is
\[
	\Rr^*(p) = \Rr_1(p) \cap \Rr_2(p).
\]
\end{thm}

Before we prove the theorem, we point out a few important 
properties of the random coding optimal rate region.

\begin{remark}[MAC form]  \label{rem:2dmic_MACform}
Let $\Rr_{1,\text{IAN}}(p)$ be the set of rate pairs $(R_1,R_2)$ such that
\begin{equation*} \label{eq:2dmic-ian}
		R_1 \le I(X_1; Y_1 \cond Q),
\end{equation*}
that is, the achievable rate (region) for the point-to-point channel $p(y_1|x_1)$
by treating the interfering signal $X_2$ as noise. Let $\Rr_{1,\text{SD}}(p)$ be the set of 
rate pairs $(R_1,R_2)$ such that 
\begin{align*}
		R_1 &\le I(X_1;Y_1 \cond X_2,Q), \\
		R_2 &\le I(X_2;Y_1 \cond X_1,Q), \\
		R_1+R_2 &\le I(X_1,X_2; Y_1 \cond Q), 
\end{align*}
that is, the achievable rate region for the multiple access channel $p(y_1|x_1,x_2)$
by decoding for both messages $M_1$ and $M_2$ simultaneously. Then,
we can express $\Rr_1(p)$ as
\begin{equation*} \label{eq:2mac-form}
		\Rr_1(p) = \Rr_{1,\text{IAN}}(p) \cup \Rr_{1,\text{SD}}(p),
\end{equation*}
which is referred to as the \emph{MAC form} of $\Rr_1(p)$, since it is the union of the achievable rate regions of 1-sender 
and 2-sender multiple access channels.
The region $\Rr_2(p)$ can be expressed similarly as the union of the \emph{interference-as-noise
region} $\Rr_{2,\text{IAN}}(p)$ and the \emph{simultaneous-decoding region} $\Rr_{2,\text{SD}}(p)$.
Hence the optimal rate region $\Rr^*(p)$ can be expressed as
\begin{align}
	\Rr^*(p)&=\bigl(\Rr_{1,\text{IAN}}(p) \cap \Rr_{2,\text{IAN}}(p) \bigr) 
	 	\cup \bigl(\Rr_{1,\text{IAN}}(p) \cap \Rr_{2,\text{SD}}(p) \bigr) \notag \\
	& \qquad \cup \bigl(\Rr_{1,\text{SD}}(p) \cap \Rr_{2,\text{IAN}}(p) \bigr)
	 	\cup \bigl(\Rr_{1,\text{SD}}(p) \cap \Rr_{2,\text{SD}}(p) \bigr),
	\label{eq:IAN_SD_expansion}
\end{align}
which is achieved by taking the union over all possible combinations of treating interference as noise and simultaneous decoding at the two receivers.
\end{remark}

\begin{remark}[Min form] \label{rem:2dmic_minform}
The region $\Rr_1(p)$ in \eqref{eq:2dmic-r1} can be equivalently
characterized as
the set of rate pairs $(R_1,R_2)$ such that 
\begin{subequations}
\begin{gather}
	R_1 \le I(X_1; Y_1 \cond X_2, Q), \\
	R_1 + \min\{R_2, I(X_2;Y_1 \cond X_1,Q) \} 
	\le I(X_1,X_2; Y_1 \cond Q).  \label{eq:2dmic_R1minform:second}
\end{gather}
\label{eq:2dmic_R1minform}%
\end{subequations}
The minimum term in \eqref{eq:2dmic_R1minform:second} can be interpreted as the effective rate of the interfering signal $X_2$ 
at the receiver $Y_1$, which is a monotone increasing function of $R_2$ and saturates at the maximum possible rate for distinguishing $X_2$ codewords; see \cite{Bandemer:2010:InterferenceDecodingJournal}. 
When $R_2$ is small, all $X_2$ codewords are distinguishable and the effective rate equals the actual code rate. 
In comparison, when $R_2$ is large, the codewords are not distinguishable and the effective rate equals $I(X_2;Y_1\cond X_1,Q)$, which is the maximum achievable rate for the channel from $X_2$ to $Y_1$.
\end{remark}
\begin{remark}[Nonconvexity]
The random coding optimal rate region $\Rr^*(p)$ is not convex in general. This is exemplified by the deterministic 2-DM-IC in Figure~\ref{fig:2dmic_nonconvex}.
\end{remark}

\begin{figure}[h!]
	\centering
	\phantom{i}\hfill
	\subfigure[Channel block diagram.]{
		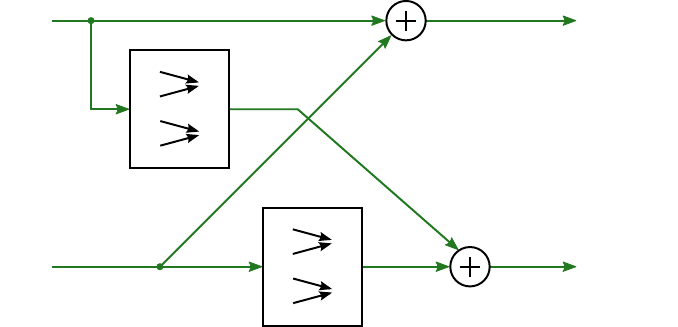
		\label{fig:2dmic_nonconvexChannel}
	}
	\hfill
	\subfigure[Regions $\Rr_1(p)$, $\Rr_2(p)$, and $\Rr^*(p)$ for $Q=\emptyset$ and $X_1,X_2 \sim \U \natSetFromZero 3$.]{
		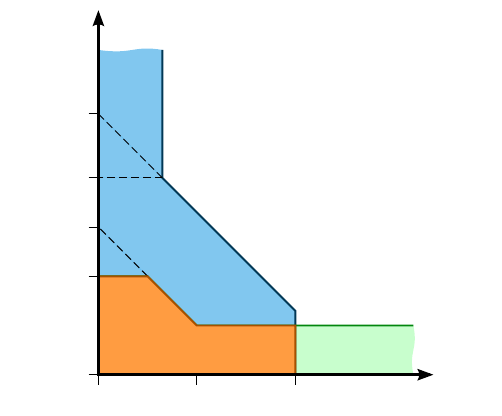
		\label{fig:2dmic_nonconvexRegion_updated}
	}
	\hfill\phantom{i}
	\caption{An example for nonconvex $\Rr^*(p)$.}
	\label{fig:2dmic_nonconvex}
\end{figure}
A direct approach to proving Theorem~\ref{thm:2dmic} would be to analyze
the performance of maximum likelihood decoding:
\begin{align*}
	\mh_1 &= \arg \max_{m_1} \frac{1}{2^{nR_2}} \sum_{m_2} 
				\prod_{i=1}^n p_{Y_1|X_1,X_2}(y_{1i} \cond x_{1i}(m_1), x_{2i}(m_2)),\\
	\mh_2 &= \arg \max_{m_2} \frac{1}{2^{nR_1}} \sum_{m_1} 
				\prod_{i=1}^n p_{Y_2|X_1,X_2}(y_{2i} \cond x_{1i}(m_1), x_{2i}(m_2))
\end{align*}
for the $p$-distributed random code.
Instead of performing this analysis, which is quite complicated (if possible),
we establish the achievability of $\Rr^*(p)$ by the suboptimal simultaneous nonunique decoding rule, which uses the notion of joint typicality. We then show that if the average probability of error
of the $(2^{nR_1}, 2^{nR_2},n; p)$ random code ensemble tends to zero as $n \to \infty$, then the rate pair
$(R_1,R_2)$ must lie in $\Rr^*(p)$.

\subsection{Proof of Achievability}  \label{sec:2dmic_achievability}

Each receiver uses \emph{simultaneous nonunique decoding}. 
Receiver~1 declares that $\mh_1$ is sent if it is the unique message among $\natSet{2^{nR_1}}$ such that
\begin{align*}
	\bigl( q^n, x_1^n(\mh_1), x_2^n(m_2), y_1^n \bigr) \in \aep \quad \text{for some $m_2 \in \natSet{2^{nR_2}}$}.
\end{align*}
If there is no such message or more than one, it declares an error.
Similarly, receiver~2 finds the unique message $\mh_2 \in \natSet{2^{nR_2}}$ such that
\[
	\bigl( q^n, x_1^n(m_1), x_2^n(\mh_2), y_2^n \bigr) \in \aep \quad \text{for some $m_1 \in \natSet{2^{nR_1}}$}.
\]
To analyze the probability of decoding error averaged over the random codebook ensemble, assume without loss of generality that $(M_1,M_2)=(1,1)$ is sent. Receiver~1 makes an error only if one or both 
of the following events occur:
\begin{align*}
	\Er_1 &= \bigl\{ ( Q^n, X_1^n(1), X_2^n(1), Y_1^n ) \notin \aep \bigr\}, \\
	\Er_2 &= \bigl\{ ( Q^n, X_1^n(m_1), X_2^n(m_2), Y_1^n ) \in \aep 
			 \text{ for some $m_1\ne 1$ and some $m_2$}  \bigr\}.
\end{align*}
By the law of large numbers, $\P(\Er_1)$ tends to zero as $n\to\infty$.

We bound $\P(\Er_2)$ in two ways, which leads to the MAC form of $\Rr_1(p)$
in Remark~\ref{rem:2dmic_MACform}.
First, since the joint typicality of the quadruple $(Q^n, X_1^n(m_1), X_2^n(m_2), Y_1^n)$ for each $m_2$ implies the joint typicality of the triple $(Q^n, X_1^n(m_1), Y_1^n)$, we have
\begin{align*}
	\Er_2 
	&\subseteq \bigl\{ ( Q^n, X_1^n(m_1), Y_1^n ) \in \aep 
			\text{ for some $m_1 \ne 1$}  \bigr\} = \Ec_2'.
\end{align*}
Then, by the packing lemma in~\cite[Section 3.2]{ElGamalKim}, $\P(\Er_2')$ 
tends to zero as $n \to \infty$ if
\begin{equation}
	R_1 < I(X_1;Y_1 \cond Q) - \delta(\eps).
	\label{eq:2dmic_ach_1}
\end{equation}
The second way to bound $\P(\Er_2)$ is to partition $\Ec_2$ into the two events
\begin{align*}
	\Er_{21} &= \bigl\{ ( Q^n, X_1^n(m_1), X_2^n(1), Y_1^n ) \in \aep \text{ for some $m_1\ne 1$}  \bigr\}, \\
	\Er_{22} &= \bigl\{ ( Q^n, X_1^n(m_1), X_2^n(m_2), Y_1^n ) \in \aep \text{ for some $m_1\ne 1$, $m_2\ne 1$}  \bigr\}.
\end{align*}
Again by the packing lemma, 
$\P(\Er_{21})$ and $\P(\Er_{22})$ tend to zero as $n\to \infty$ if
\begin{subequations}
	\begin{align}
	R_1 &< I(X_1;Y_1 \cond X_2,Q) - \delta(\eps), \\
	R_1+R_2 &< I(X_1,X_2; Y_1 \cond Q)- \delta(\eps).
	\end{align}
	\label{eq:2dmic_ach_2}%
\end{subequations}
Thus we have shown that the average probability of decoding error at receiver~1 tends
to zero as $n \to \infty$ if at least one of~\eqref{eq:2dmic_ach_1} and~\eqref{eq:2dmic_ach_2} holds. 
Similarly, we can show that the average probability of decoding error at receiver~2 tends
to zero as $n \to \infty$ if $R_2 < I(X_2;Y_2 \cond Q) - \delta(\eps)$, or 
$R_2 < I(X_2;Y_2 \cond X_1,Q) - \delta(\eps)$ and
$R_1+R_2 < I(X_1,X_2; Y_2 \cond Q)- \delta(\eps)$.
Since $\eps > 0$ is arbitrary and $\delta(\eps) \to 0$ as $\eps \to 0$,
this completes the proof of achievability for any rate pair
$(R_1,R_2)$ in the interior of $\Rr_1(p) \cap \Rr_2(p)$.\hfill\qed

\begin{remark}[Comparison to maximum likelihood decoding]
It is instructive to consider the following progression of decoding rules for receiver~1.

\begin{enumerate}
\item \label{it:rule_ML} Maximum likelihood decoding: 
	\begin{align}
		\mh_1 
		&= \arg \max_{m_1} p(y_1^n \cond m_1) \notag\\
		&= \arg \max_{m_1} \frac{1}{2^{nR_2}} \sum_{m_2} p(y_1^n \cond m_1,m_2) 
			\label{eq:mld}\\
		&= \arg \max_{m_1} \frac{1}{2^{nR_2}} \sum_{m_2} 
				\prod_{i=1}^n p_{Y_1|X_1,X_2}(y_{1i} \cond x_{1i}(m_1), x_{2i}(m_2)), \nonumber
	\end{align}
which is the optimal decoding rule.	

\item \label{it:rule_SML} Simultaneous maximum likelihood decoding:
	\[
		\mh_1 = \arg \max_{m_1} \max_{m_2} p(y_1^n \cond m_1,m_2),
	\]
which is equivalent to performing optimal decoding of the message pair $(M_1,M_2)$ and then taking the first
coordinate. Note the maximum over $m_2$ instead of the average as in~\eqref{eq:mld}.

\item \label{it:rule_Typscore} Typicality score decoding:
	\[
		\mh_1 = \arg \min_{m_1} \min_{m_2} \eps^\star(y_1^n, m_1, m_2),
	\]
	where $\eps^\star(y_1^n, m_1, m_2)$ is defined as the smallest $\eps$ such that
	\[
		(q^n, x_1^n(m_1), x_2^n(m_2), y_1^n) \in \aep.
	\]
Here the notion of joint typicality plays the role of likelihood in previous decoding rules
and $\eps^\star$ captures the penalty for being atypical.

\item \label{it:rule_SND} Simultaneous nonunique decoding: Find the unique $\mh_1$ such that
	\[
		(q^n, x_1^n(\mh_1), x_2^n(m_2), y_1^n) \in \aep \quad \text{for some $m_2$}.
	\]
This is equivalent to performing typicality score decoding with predetermined threshold $\eps$ for 
$\eps^\star(y_1^n, m_1, m_2)$; thus first forming a list of all $(m_1,m_2)$ for which $\eps^\star(y_1^n, m_1, m_2) \le \eps$, and then taking the first coordinate of the members of the list (if it is unique).
\end{enumerate}
Starting from the optimal maximum likelihood decoding rule,
each subsequent rule modifies its predecessor by ``relaxing'' one step. Nonetheless, 
these relaxation steps do not result in any significant loss in performance, as 
is evident in the rate-optimality of the simultaneous nonunique decoding rule.
\end{remark}

\begin{remark}
As observed in~\cite{Bidokhti2012} (see also \eqref{eq:IAN_SD_expansion} in Remark~\ref{rem:2dmic_MACform} above), 
each rate point in $\Rr^*(p)$ can alternatively be achieved by having each receiver specifically decode for
either the desired message alone or both the desired and interfering messages.
\end{remark}

\subsection{Proof of the Converse} \label{sec:2dmic_converse}
Fix a pmf $p = p(q)\inpmfspace p(x_1|q)\inpmfspace p(x_2|q)$ and let $(R_1,R_2)$ be a rate pair achievable by the $p$-distributed random code ensemble.
We prove that this implies that $(R_1,R_2) \in \Rr_1(p) \cap \Rr_2(p)$ as claimed.
Here, we show the details for the inclusion $(R_1,R_2) \in \Rr_1(p)$; the proof for $(R_1,R_2) \in \Rr_2(p)$ follows similarly.
With slight abuse of notation, let $\Cc_n$ denote the random codebook (and the time sharing sequence), namely $(Q^n, X_1^n(1), \dots, X_1^n(2^{nR_1}), X_2^n(1), \dots, X_2^n(2^{nR_2}))$. 

First consider a fixed codebook $\Cc_n=c$. By Fano's inequality,
\begin{align*}
	H(M_1 \cond Y_1^n, \Cc_n=c) &\le 1 + nR_1 \pen(c).
\end{align*}
Taking the expectation over the random codebook $\Cc_n$, it follows that
\begin{align}
	H(M_1 \cond Y_1^n, \Cc_n) &\le 1 + nR_1 \E_{\Cc_n}[ \pen(\Cc_n) ]  
	\le n\eps_n,
	\label{eq:2dmic_fano}
\end{align}
where $\eps_n \to 0$ as $n\to\infty$ since $\E_{\Cc_n}[ \pen(\Cc_n) ] \to 0$.

We prove the conditions in the min form~\eqref{eq:2dmic_R1minform}. To see that the first inequality is true, note that
{\allowdisplaybreaks
\begin{align*}
n( R_1 - \eps_n)
	&= H(M_1 \cond \Cc_n) - n\eps_n \\
	&\annleq{a} I(M_1; Y_1^n \cond \Cc_n)  \\
	&\le I(X_1^n; Y_1^n \cond \Cc_n) \\
	&\le I(X_1^n; Y_1^n, X_2^n \cond \Cc_n) \\
	&= I(X_1^n; Y_1^n\cond X_2^n, \Cc_n ) \\
	&= H(Y_1^n \cond X_2^n, \Cc_n) - H(Y_1^n \cond X_1^n, X_2^n, \Cc_n)  \\
	&\annleq{b} H(Y_1^n \cond X_2^n, Q^n) - H(Y_1^n \cond X_1^n, X_2^n, Q^n)  \\
	&\anneq{c} nI(X_1; Y_1 \cond X_2, Q),
\end{align*}}%
where (a) follows by (the averaged version of) Fano's inequality in~\eqref{eq:2dmic_fano}, (b) follows by omitting some conditioning and using the memoryless property of the channel, and (c) follows since the tuple $(Q_i,X_{1i},X_{2i},Y_i)$ is i.i.d.\@ for all~$i$.
Note that unlike conventional converse proofs where nothing can be assumed about the codebook structure, here we can take advantage of the properties of a given codebook generation procedure.

To prove the second inequality in~\eqref{eq:2dmic_R1minform}, we need the following lemma, which is proved in Appendix~\ref{sec:proof_2MAC_singleLetterizeEntropy}.
\begin{lemma}  \label{thm:2MAC_singleLetterizeEntropy}
	\begin{align*}
		\lim_{n\to \infty} \frac{1}{n} H(Y_1^n \cond X_1^n, \Cc_n) 
		= H(Y_1\cond X_1,X_2,Q) + \min\{ R_2,I(X_2;Y_1\cond X_1,Q)\}.
	\end{align*}
\end{lemma}
The lemma states that depending on $R_2$, $(1/n) H(Y_1^n \cond X_1^n, \Cc_n)$ either tends to $H(Y_1\cond X_1,Q)$, that is, the remaining received sequence after recovering the desired codeword looks like i.i.d.~noise, or to $R_2 + H(Y_1\cond X_1,X_2,Q)$, 
that is, the receiver can distinguish the interfering codeword from the noise.

\medskip

Equipped with this lemma, we have
\begin{align*}
	n(R_1-\eps_n) 
	&\annleq{a} I(X_1^n; Y_1^n \cond \Cc_n) \\
	&= H(Y_1^n\cond \Cc_n) - H(Y_1^n\cond X_1^n,\Cc_n) \\
	&\le H(Y_1^n\cond Q^n) - H(Y_1^n\cond X_1^n,\Cc_n) \\
	&\annleq{b} nH(Y_1\cond Q) - nH(Y_1\cond X_1,X_2,Q) - \min\{ nR_2,nI(X_2;Y_1\cond X_1,Q)\} + n\eps_n \\
	&= nI(X_1,X_2;Y_1\cond Q) + \min\{nR_2, nI(X_2;Y_1\cond X_1,Q)\} + n\eps_n.
\end{align*}
Here, (a) follows by Fano's inequality and (b) follows by Lemma~\ref{thm:2MAC_singleLetterizeEntropy} with some $\eps_n$ that tends to zero as $n \to \infty$. The conditions for $\Rr_2(p)$ can be proved similarly.
This completes the proof of the converse.

\section{DM-IN with $K$~Senders and $L$~Receivers} \label{sec:KLdmic}

We generalize the previous result to the $K$-sender, $L$-receiver discrete memoryless interference network ($(K,L)$-DM-IN) with input alphabets $\Xc_1,\dots,\Xc_K$, output alphabets $\Yc_1,\dots,\Yc_L$, and pmfs $p(y_1,\dots,y_L\cond x_1,\dots,x_K)$. In this network, each sender $k \in \natSet K$ communicates an independent message $M_k$ at rate $R_k$ and each receiver $l \in \natSet L$ wishes to recover the messages sent by  a subset $\Dc_l \subseteq \natSet K$ of senders (also referred to as a demand set). The channel is depicted in Figure~\ref{fig:KLdmic}.

\begin{figure}[h]
	\centering
	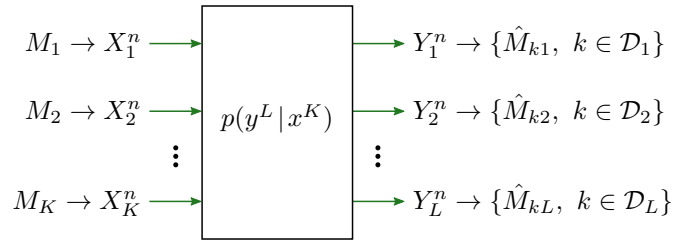
	\caption{Discrete memoryless interference network with $K$ senders and $L$ receivers.}
	\label{fig:KLdmic}
\end{figure}

More formally, a $(2^{nR_1},\dots,2^{nR_K},n)$ code $\Cc_n$ for the $(K,L)$-DM-IN consists of
\begin{itemize}
\item
$K$ message sets $\natSet{2^{nR_1}}$, \dots, $\natSet{2^{nR_K}}$,

\item 
$K$ encoders, where encoder~$k \in \natSet K$ assigns a code\-word $x_k^n(m_k)$ to each message $m_k \in \natSet{2^{nR_k}}$,

\item
$L$ decoders, where decoder~$l \in \natSet L$ assigns estimates $\mh_{kl}$, $k\in\Dc_l$, or an error message $\error$ to each received sequence $y_l^n$.

\end{itemize}

We assume that the message tuple $(M_1,\dots,M_K)$ is uniformly distributed over $\natSet{2^{nR_1}}\times \dots \times \natSet{2^{nR_K}}$. The average probability of error for the code $\Cc_n$ is defined as
\[
\pen(\Cc_n) = \P\bigl\{ \Mh_{kl} \ne M_k \text{ for some $l \in \natSet L,\, k\in \Dc_l$} \bigr\}.
\]
A rate tuple $(R_1,\dots,R_K)$ is said to be \emph{achievable} for the DM-IN if there exists a sequence of $(2^{nR_1}, \dots, 2^{nR_K}, n)$ codes $\Cc_n$ such that $\lim_{n \to \infty} \pen(\Cc_n) = 0$. The \emph{capacity region} $\Cr$ of the $(K,L)$-DM-IN is the closure of the set of achievable rate tuples $(R_1, \dots, R_K)$.

As in Section~\ref{sec:2dmic}, we limit our attention to a randomly generated code ensemble with a special structure. Let $p = p(q)\inpmfspace p(x_1|q)\cdots p(x_K|q)$ be a given pmf on $\Qc \times \Xc_1 \times \dots \times \Xc_K$, where $\Qc$ is a finite alphabet. Suppose that codewords $X_k^n(m_k)$, $m_k \in \natSet{2^{nR_k}}$, $k\in \natSet K$, are generated randomly as follows.
\begin{itemize}
\item Let $Q^n \sim \prod_{i=1}^n p_Q(q_i)$. 
\item For each $k\in \natSet K$ and $m_k \in \natSet{2^{nR_k}}$, let $X_k^n(m_k) \sim \prod_{i=1}^n p_{X_k|Q}(x_{ki}|q_i)$,
conditionally independent given $Q^n$.
\end{itemize}
Each instance of codebooks generated in this manner, along with the corresponding optimal decoders, constitutes
a $(2^{nR_1}, \dots, 2^{nR_K}, n)$ code. We refer to the random code ensemble thus generated
as the $(2^{nR_1}, \dots, 2^{nR_K}, n; p)$ \emph{random code ensemble}.

\begin{definition}[Random coding optimal rate region]  \label{def:KLdmic_randomcoderegion}
Given a pmf $p = p(q)\inpmfspace p(x_1|q) \cdots p(x_K|q)$,
the \emph{optimal rate region $\Rr^*(p)$ achievable by the $p$-distributed random
code ensemble} is the closure of the set of rate tuples $(R_1,\dots,R_K)$ such that
the sequence of the $(2^{nR_1},\dots,2^{nR_K},n; p)$ random code ensembles $\Cc_n$ satisfies
\[
\lim_{n\to\infty} \E_{\Cc_n}[ \pen(\Cc_n) ]= 0,
\] 
where the expectation is with respect to the random code ensemble $\Cc_n$.
\end{definition}

Note that the setup discussed in Section~\ref{sec:2dmic} as well as the broadcast channel example in Subsection~\ref{sec:intro_2bc} correspond to the special case of $K=L=2$ and demand sets $\Dc_1 = \{1\}$ and $\Dc_2 =\{2\}$. 
More generally, the $p$-distributed random code ensemble for the $(K,L)$-DM-IN captures superposition coding with an arbitrary number of layers. Suppose that there are $K$ senders, some of which need to communicate multiple messages (see Figure~\ref{fig:superposToKL1}). In superposition coding, each message at a sender is encoded into a codeword $U_{k'}^n$ and the sender 
combines (superimposes) all such codewords. By merging the combining functions at the sender with the physical channel $p(y^L | x^K)$,
we obtain a $(K',L)$-DM-IN $p(y^L | u^{K'})$ with ``virtual'' inputs $U_{k'}$, $k' \in \natSet{K'}$, as illustrated in Figure~\ref{fig:superposToKL2}.%
\begin{figure}[h]
	\centering
	\subfigure[Multiple messages per sender via superposition coding.]{
		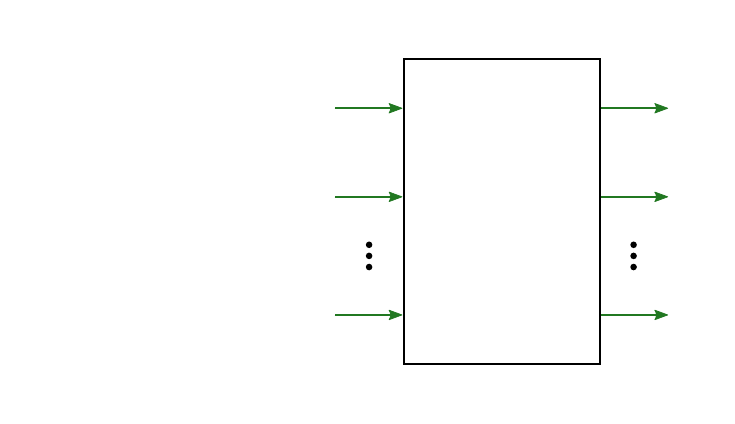
		\label{fig:superposToKL1}
	}
	\hfill
	\subfigure[Equivalent channel with a single message per sender.]{
		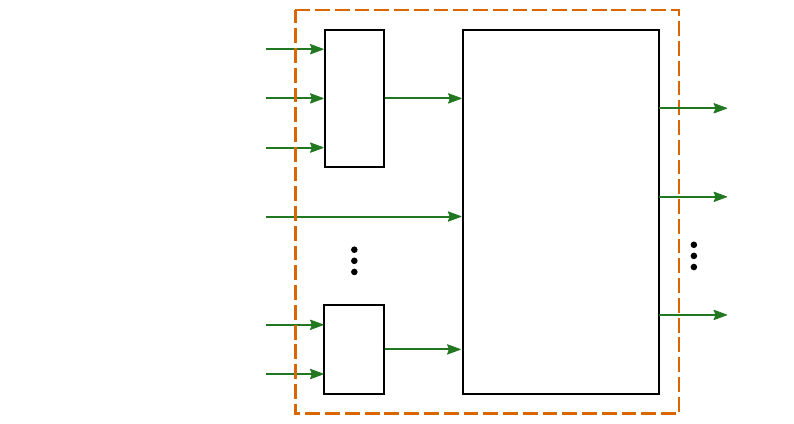
		\label{fig:superposToKL2}
	}
	\caption{The class of $(K,L)$-DM-INs includes superposition coding with an arbitrary number of layers.}
	\label{fig:superposToKL}
\end{figure}%

Define the rate region $\Rr_1(p)$ as
\begin{equation}
	\Rr_1(p) = \bigcup_{\substack{\Sc \subseteq \natSet K,\\ \Dc_1 \subseteq \Sc}} \Rr_{\text{MAC}(\Sc)}(p),
	\label{eq:KLdmic_R1macform}
\end{equation}
where $\Rr_{\text{MAC}(\Sc)}(p)$ is the achievable rate region for the multiple access channel from the set of senders $\Sc$ to receiver~1, i.e., the set of rate tuples $(R_1,\dots,R_K)$ such that
\begin{align*}
	R_{\Tc} = \sum_{j \in \Tc} R_j & \le I(X_{\Tc}; Y_1 \cond X_{\Sc\setminus\Tc}, Q) \quad \text{for all $\Tc\subseteq \Sc$}.
\end{align*}
Note that the set $\Rr_{\text{MAC}(\Sc)}(p)$ corresponds to the rate region achievable by decoding for the messages from the senders $\Sc$, which contains all desired messages and possibly some interfering messages. Also note that $\Rr_{\text{MAC}(\Sc)}(p)$ contains upper bounds
only on the rates $R_k$, $k\in \Sc$, of the active senders $\Sc$ in the MAC. The signals from the inactive senders in $\compl{\Sc}$ are treated as noise
and the corresponding rates $R_k$ for $k\in \compl{\Sc}$ are unconstrained. Consequently, $\Rr_1(p)$ is unbounded in the coordinates $R_k$ for $k \in \natSet K \setminus \Dc_1$.

The region $\Rr_1(p)$ 
in \eqref{eq:KLdmic_R1macform}
can equivalently be written as the set of rate tuples $(R_1,\dots,R_K)$ such that for all $\Uc \subseteq \natSet K \setminus \Dc_1$ and for all $\Dc$ with $\emptyset \subset \Dc \subseteq \Dc_1 $, 
\begin{equation}
	R_{\Dc} + \min_{\Uc'\subseteq \Uc} \bigl\{R_{\Uc'} + I(X_{\Uc\setminus\Uc'}; Y_1 \cond X_{\Dc},X_{\Uc'},X_{\natSet{K}\setminus \Dc \setminus\Uc},Q) \bigr\}
	\le I(X_{\Dc}, X_{\Uc}; Y_1 \cond X_{\natSet{K}\setminus \Dc \setminus\Uc},Q ).
	\label{eq:KLdmic_R1minform}
\end{equation}
As in the case of the 2-DM-IC, each argument of each term in the minimum represents a different mode of signal saturation.
The equivalence between the MAC form~\eqref{eq:KLdmic_R1macform} and the min form~\eqref{eq:KLdmic_R1minform} can be proved by identifying the largest set of decodable interfering messages as in~\cite{Motahari2011}. For completeness, we provide a proof in Appendix~\ref{sec:KLdmic_minMacFormEquivalence}.

\begin{remark}
The MAC and min forms of $\Rr_1(p)$ are duals to each other in the following sense. The condition for $(R_1,\dots,R_K) \in \Rr_1(p)$ in the MAC form~\eqref{eq:KLdmic_R1macform} can be expressed as
\begin{align}
	&\exists \Sc \subseteq \natSet K, \quad \Dc_1 \subseteq \Sc: \notag \\
	& \hspace{15mm} \forall \Tc \subseteq \Sc: \notag \\
	& \hspace{30mm} R_{\Tc} \le I(X_{\Tc}; Y_1 \cond X_{\Sc\setminus\Tc}, Q). 
	\label{eq:KLdmic_R1macformCond} 
	\\
\intertext{The conditions in the min form~\eqref{eq:KLdmic_R1minform} can be rewritten{\footnotemark} as}
	& \forall \Vc \subseteq \natSet K, \quad \Vc \cap \Dc_1 \ne \emptyset : \notag \\
	& \hspace{15mm} \exists \Vc' \subseteq \Vc, \quad \Vc'\cap\Dc_1 = \Vc\cap\Dc_1: \notag \\ 
	& \hspace{30mm} R_{\Vc'} \le I(X_{\Vc'}; Y_1 \cond X_{\natSet{K}\setminus\Vc},Q).
	\label{eq:KLdmic_R1minformCond}
\end{align}
\footnotetext{To see this, first note that the minimum terms on the left hand side of~\eqref{eq:KLdmic_R1minform} represent a set of conditions of which at least one has to be true, then use the identity 
\begin{align*}
	I(X_{\Dc}, X_{\Uc}; Y_1 \cond X_{\natSet{K}\setminus \Dc \setminus\Uc},Q ) 
	- I(X_{\Uc\setminus\Uc'}; Y_1 \cond X_{\Dc},X_{\Uc'},X_{\natSet{K}\setminus \Dc \setminus \Uc},Q) 
	&
	= I(X_{\Dc},X_{\Uc'}; Y_1 \cond X_{\natSet{K}\setminus \Dc \setminus\Uc},Q),
\end{align*}
and finally, let $\Vc = \Uc \cup \Dc$ and $\Vc' = \Uc' \cup \Dc$.}%
Both conditions involve a set of messages from the senders $\Sc$ (or $\Vc$) 
and its subset $\Tc$ (or $\Vc'$), and impose a mutual information upper bound on the sum rate over the subset. The key difference is the order of the quantifiers $\forall$ and $\exists$.
\end{remark}

Analogous to $\Rr_1(p)$, define the regions $\Rr_2(p), \dots, \Rr_L(p)$ for receivers $2,\dots,L$ by making appropriate index substitutions. 
We are now ready to state the main result for the $(K,L)$-DM-IN.
\begin{thm}  \label{thm:KLdmic}
Given a pmf $p = p(q)\inpmfspace p(x_1|q)\cdots p(x_K|q)$,
the optimal rate region of the $(K,L)$-DM-IN $p(y^L|x^K)$ with demand sets $\Dc_1, \ldots, \Dc_L$
achievable by the $p$-distributed random code ensemble is
\[
	\Rr^*(p) = \bigcap_{l \in \natSet L} \Rr_l(p).
\]
\end{thm}

Note that, as for its 2-DM-IC counterpart, this region is not convex in general.	

\begin{example} 
	Consider the $K$-user-pair Gaussian interference network 
	\begin{align*}
		Y_l &= \sum_{k=1}^K g_{kl} X_k + Z_l, \quad l \in \natSet K, 
	\end{align*}
	where $Z_l \sim \N(0,1)$ and $g_{kl}$ are channel gains from sender $k$ to receiver $l$. 
	Assume the Gaussian random code ensemble with $X_k \sim\N(0,1)$, $k\in\natSet K$. The optimal rate region achievable by this random code ensemble was established in~\cite{Motahari2011} and~\cite{Baccelli2011}, and can be recovered from Theorem~\ref{thm:KLdmic} by letting $K=L$, $\Dc_k = \{k\}$ for $k\in\natSet K$, and applying the discretization procedure in~\cite[Section 3.4]{ElGamalKim}. Theorem~\ref{thm:KLdmic} generalizes this result in several directions, since (a) it applies to non-Gaussian networks, (b) it applies to non-Gaussian random code ensembles (which is crucial to analyze the performance under a fixed constellation), and (c) it includes coded time sharing and superposition coding.
\end{example}

\begin{example}
	Consider the deterministic interference channel with three sender--receiver pairs (3-DIC)~\cite{Bandemer:2010:InterferenceDecodingJournal}, where 
	\begin{align*}
		Y_1 &= f_1( g_{11}(X_1), h_1( g_{21}(X_2), g_{31}(X_3) ), \\
		Y_2 &= f_2( g_{22}(X_2), h_2( g_{32}(X_3), g_{12}(X_1) ), \\
		Y_3 &= f_3( g_{33}(X_3), h_3( g_{13}(X_1), g_{23}(X_2) )
	\end{align*}
	for some loss functions $g_{kl}$ and combining functions $h_k$ and $f_k$,  $k,l\in \{1,2,3\}$. The combining functions are supposed to be injective in each argument. This setting is of interest since it contains as special cases the El~Gamal--Costa two-user-pair interference channel~\cite{ElGamalCosta82}, for which the Han--Kobayashi coding scheme achieves the capacity region, and the Avestimehr--Diggavi--Tse $q$-ary expansion deterministic (QED) interference channel~\cite{Avestimehr2011}, which approximates Gaussian interference networks in the high-power regime.
	The 3-DIC is an instance of a $(K,L)$-DM-IN with $L=K=3$ and $\Dc_k = \{k\}$ for $k\in \natSet K$.
	The interference decoding inner bound on the 3-DIC capacity region in~\cite{Bandemer:2010:InterferenceDecodingJournal} coincides with the region in Theorem~\ref{thm:KLdmic} in its min form. Beyond the results in~\cite{Bandemer:2010:InterferenceDecodingJournal}, Theorem~\ref{thm:KLdmic} establishes that the interference decoding inner bound is in fact optimal given the codebook structure.
	Note that for the 3-DIC channel, we can identify each minimum term with a specific signal in the channel block diagram for which the term counts the number of distinguishable sequences.
\end{example}

\begin{proof}[Proof of Theorem~\ref{thm:KLdmic}]
We focus only on receiver 1 for which $M_k$, $k\in \Dc_1$, are the desired messages and $M_k$, $k \in \compl{\Dc_1} = \natSet K \setminus \Dc_1$, are interfering messages.
Achievability is proved using simultaneous nonunique decoding. Receiver~1 declares that $\mh_{\Dc_1}$ is sent if it is the unique message tuple such that 
\begin{align*}
	\bigl( q^n, x^n_{\Dc_1}(\mh_{\Dc_1}), x^n_{\compl{\Dc_1}}(m_{\compl{\Dc_1}}), y_1^n \bigr) \in \aep \quad \text{for some $m_{\compl{\Dc_1}}$},
\end{align*}
where $x^n_{\Dc_1}(\mh_{\Dc_1})$ is the tuple of $x_k^n(\mh_k)$, $k \in \Dc_1$, and similarly, $x^n_{\compl{\Dc_1}}(m_{\compl{\Dc_1}})$ is the tuple of $x_k^n(m_k)$, $k \in \compl{\Dc_1}$.
The analysis follows similar steps as in Subsection~\ref{sec:2dmic_achievability}.

To prove the converse, fix a pmf $p$ and let $(R_1,\dots,R_K)$ be a rate tuple that is achievable by the $p$-distributed random code ensemble. 
We need the following generalization of Lemma~\ref{thm:2MAC_singleLetterizeEntropy}, which is proved in Appendix~\ref{sec:proof_KLMAC_singleLetterizeEntropy}.

\begin{lemma}  \label{thm:KLMAC_singleLetterizeEntropy}
	If $\Dc_1 \subseteq \Sc \subseteq \natSet K$, then
	\begin{align*}
		\lim_{n\to \infty} \frac{1}{n} H(Y_1^n \cond X_{\Sc}^n, \Cc_n) 
		& = H(Y_1\cond X_{\natSet K},Q) 
		  + \min_{\Uc \subseteq \compl{\Sc}}( R_{\Uc} + I(X_{\compl{(\Sc \cup \Uc)}};Y_1\cond X_{\Sc \cup \Uc},Q)).
	\end{align*}
\end{lemma}
We now establish~\eqref{eq:KLdmic_R1minform} as follows. Fix a subset of desired message indices, $\Dc \subseteq \Dc_1$, and a subset of interfering message indices, $\Uc \subseteq \compl{\Dc_1}$. Then
\begin{align*}
	 n(R_{\Dc}-\eps_n)  
	&\annleq{a} I(X_{\Dc}^n; Y_1^n \cond \Cc_n) \\
	&\le I(X_{\Dc}^n; Y_1^n, X^n_{\compl{(\Dc \cup \Uc)}}\cond \Cc_n) \\
	&\le I(X_{\Dc}^n; Y_1^n \cond X^n_{\compl{(\Dc \cup \Uc)}}, \Cc_n) \\
	&= H(Y_1^n\cond X^n_{\compl{(\Dc \cup \Uc)}}, \Cc_n) - H(Y_1^n\cond X^n_{\compl{\Uc}}, \Cc_n) \\
	&\annleq{b} nH( Y_1\cond X_{\compl{(\Dc \cup \Uc)}}, Q ) - nH(Y_1\cond X_{\natSet K},Q) 
		- n \!\cdot\! \min_{\Uc' \subseteq \Uc} ( R_{\Uc'} + I(X_{\compl{(\compl{\Uc} \cup \Uc')}};Y_1\cond X_{\compl{\Uc} \cup \Uc'},Q) ) + n\eps_n \\
	&= nI( X_{\Dc\cup\Uc}; Y_1^n \cond X_{\compl{(\Dc \cup \Uc)}}, Q ) 
		- n \!\cdot\! \min_{\Uc' \subseteq \Uc} ( R_{\Uc'} + I(X_{\Uc \setminus \Uc'};Y_1\cond X_{\compl{(\Uc\setminus \Uc')}},Q) ) + n\eps_n,
\end{align*}
where (a) follows by Fano's inequality and (b) follows by Lemma~\ref{thm:KLMAC_singleLetterizeEntropy}. 
This completes the proof of the converse.
\end{proof}

\section{Application to the Han--Kobayashi coding scheme} \label{sec:HK}
We revisit the two-user-pair DM-IC in Figure~\ref{fig:2dmic}. The best known inner bound on the capacity region is achieved by the Han--Kobayashi coding scheme~\cite{HanKobayashi81}. In this scheme, the message $M_1$ is split into common and private messages $M_{12}$ and $M_{11}$ at rates $R_{12}$ and $R_{11}$, respectively, such that $R_1= R_{12} + R_{11}$. Similarly $M_2$ is split into common and private messages $M_{21}$ and $M_{22}$ at rates $R_{21}$ and $R_{22}$ such that $R_2=R_{22} + R_{21}$. More specifically, the scheme uses random codebook generation and coded time sharing as follows. 
Fix a pmf $p=p(q) \inpmfspace p(u_{11}|q) \inpmfspace p(u_{12}|q) \inpmfspace p(u_{21}|q) \inpmfspace p(u_{22}|q)\inpmfspace p(x_1|u_{11},u_{12},q)\inpmfspace p(x_2|u_{21},u_{22},q)$, where the latter two conditional pmfs represent deterministic mappings $x_1(u_{11},u_{12})$ and $x_2(u_{21},u_{22})$.  
Randomly generate a coded time sharing sequence $q^n \sim$ $ \prod_{i=1}^n p_Q(q_i)$. For each $k,k' \in \{1,2\}$ and $m_{kk'} \in \natSet{2^{nR_{kk'}}}$, randomly and conditionally independently generate a sequence $u_{kk'}^n(m_{kk'})$ according to $\prod_{i=1}^n p_{U_{kk'}|Q}(u_{kk'i}|q_i)$. To communicate message pair $(m_{11},m_{12})$, sender $1$ transmits $x_{1i}=x_1( u_{11i}, u_{12i})$ for $i \in \natSet n$, and analogously for  sender 2.
Receiver $k = 1,2$ recovers its intended message $M_k$ and the common message from the other sender  (although it is not required to). While this decoding scheme helps reduce the effect of interference, it results in additional constraints on the rates for common messages. 
The Han--Kobayashi coding scheme is illustrated in Figure~\ref{fig:HK42dmic_HKdecod}.

\begin{figure}[h]
	\centering
	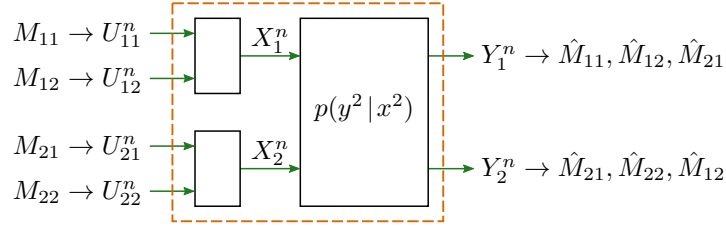
	\caption{Han--Kobayashi coding scheme.}
	\label{fig:HK42dmic_HKdecod}
\end{figure}

Let $\Rr_{\text{HK},1}(p)$ be defined as the set of rate tuples $(R_{11},R_{12},R_{21},R_{22})$ such that
\begin{subequations}
\begin{align}
	R_{11} &\le I(U_{11}; Y_1 \cond U_{12}, U_{21}, Q), \\
	R_{12} &\le I(U_{12}; Y_1 \cond U_{11}, U_{21}, Q), \\
	R_{21} &\le I(U_{21}; Y_1 \cond U_{11}, U_{12}, Q), \\
	R_{11}+R_{12} &\le I(U_{11},U_{12}; Y_1 \cond U_{21}, Q), \\
	R_{11}+R_{21} &\le I(U_{11},U_{21}; Y_1 \cond U_{12}, Q), \\
	R_{12}+R_{21} &\le I(U_{12},U_{21}; Y_1 \cond U_{11}, Q), \\
	R_{11}+R_{12}+R_{21} &\le I(U_{11},U_{12},U_{21}; Y_1 \cond Q).	
\end{align}
\label{eq:R_HK1}%
\end{subequations}
Similarly, define $\Rr_{\text{HK},2}(p)$ by making the sender/receiver index substitutions $1\leftrightarrow 2$ in the definition of $\Rr_{\text{HK},1}(p)$. 
As shown by Han and Kobayashi~\cite{HanKobayashi81}, the coding scheme achieves any rate pair $(R_1,R_2)$ that is in the interior 
of
\begin{equation}
\Rr_\text{HK} = \ReduceA\left(\bigcup_p  \Rr_{\text{HK},1}(p) \cap \Rr_{\text{HK},2}(p)\right), \label{eq:R_HK} 
\end{equation}
where $\ReduceA$ is the projection that maps the $4$-dimensional (convex) set of rate tuples $(R_{11},R_{12},R_{21},R_{22})$ into a $2$-dimensional rate region of rate pairs $(R_1,R_2)
= (R_{11}+R_{12}, R_{21}+R_{22})$. 

We are interested in finding the rate region that is achievable by the Han--Kobayashi encoding functions in conjunction with the \emph{optimal} decoding functions. To this end, note that by combining the channel and the deterministic mappings as indicated by the dashed box in Figure~\ref{fig:HK42dmic_HKdecod}, the channel $(U_{11},U_{12},U_{21},U_{22}) \to (Y_1,Y_2)$ is a $(4,2)$-DM-IN. After removing the artificial requirement for each decoder to recover the interfering sender's common message, the message demands are $\Dc_1=\{11,12\}$ and $\Dc_2=\{21,22\}$. Moreover, the Han--Kobayashi encoding scheme is in fact the $p$-distributed random code ensemble applied to this network, as defined in Section~\ref{sec:KLdmic}. 
\begin{definition}
	The \emph{optimal rate region $\Rr_\text{opt}$ achievable by the Han--Kobayashi random code ensemble} is defined as
	\begin{align*}
		\Rr_\text{opt} &= \ReduceA\left( \bigcup_{p}  \Rr^*(p) \right),
	\end{align*}
	where the union is over pmfs of the form $p = p(q)\inpmfspace p(u_{11}|q) \inpmfspace p(u_{12}|q) \inpmfspace p(u_{21}|q) \inpmfspace p(u_{22}|q) \inpmfspace p(x_1| u_{11}, u_{12}) \inpmfspace p(x_2 | u_{21}, u_{22})$ with the latter two factors representing deterministic mappings $x_1(u_{11},u_{12})$ and $x_2(u_{21},u_{22})$, and $\Rr^*(p)$ is the optimal rate region achievable by the $p(q)\inpmfspace p(u_{11}|q) \inpmfspace p(u_{12}|q) \inpmfspace p(u_{21}|q) \inpmfspace p(u_{22}|q)$-distributed random code ensemble for the $(4,2)$-DM-IN $p(y_1,y_2|u_{11},u_{12},u_{21},u_{22}) = p_{Y_1,Y_2|X_1,X_2}(y_1,y_2|x_1(u_{11},u_{12}), x_2(u_{21},u_{22}))$ (cf.\@ Definition~\ref{def:KLdmic_randomcoderegion}).
\end{definition}

Then Theorem~\ref{thm:KLdmic} implies the following.
\begin{corollary}
  \label{thm:HK_OptimalReceiver}
$		\Rr_\text{opt} = \Rr_\text{HK}.$
\end{corollary}
The corollary states that the Han--Kobayashi inner bound is optimal when encoding is restricted to randomly generated codebooks, superposition coding, and coded time sharing. It cannot be enlarged by replacing the decoders used in the proof of~\eqref{eq:R_HK1} with optimal decoders.

\begin{proof}[Proof of Corollary~\ref{thm:HK_OptimalReceiver}]
	Applying Theorem~\ref{thm:KLdmic} to the definition of $\Rr_\text{opt}$ yields
	\begin{align*}
		\Rr_\text{opt} &= \ReduceA\left(\bigcup_{p}  \Rr_1(p) \cap \Rr_2(p) \right), 
	\end{align*}
	where $\Rr_1(p)$ is the set of rate tuples $(R_{11},R_{12},R_{21},R_{22})$ such that
	\begin{align}
		R_{\Tc_1} &\le I(U_{\Tc_1}; Y_1 \cond U_{\Sc_1\setminus\Tc_1}, Q) \quad \text{for all $\Tc_1 \subseteq \Sc_1$} \label{eq:R_opt1}
	\end{align}
	for some $\Sc_1$ with $\{11,12\} \subseteq \Sc_1 \subseteq \{11,12,21,22\}$. 
	Likewise, $\Rr_2(p)$ is the set of rate tuples that satisfy 
	\begin{align}
		R_{\Tc_2} &\le I(U_{\Tc_2}; Y_2 \cond U_{\Sc_2\setminus\Tc_2}, Q) \quad \text{for all $\Tc_2 \subseteq \Sc_2$} \label{eq:R_opt2}
	\end{align}
	for some $\Sc_2$ with $\{21,22\} \subseteq \Sc_2 \subseteq \{11,12,21,22\}$. Here, $\Sc_1$ and $\Sc_2$ contain the indices of the messages recovered by receivers~1 and~2, respectively.

	In order to compare $\Rr_\text{opt}$ to $\Rr_\text{HK}$, recall~\eqref{eq:R_HK1} and~\eqref{eq:R_HK} and the compact description of $\Rr_\text{HK}$ in~\cite{ChongMotani08} as the set of all rate pairs $(R_1,R_2)$ such that
	{\allowdisplaybreaks%
	\begin{subequations}
	\begin{align}
		R_1 &\le I(U_{11},U_{12};Y_1 \cond U_{21},Q), \\
		R_2 &\le I(U_{21},U_{22};Y_2 \cond U_{12},Q), \\
		R_1+R_2 &\le I(U_{11},U_{12},U_{21}; Y_1 \cond Q) 
			+ I(U_{22}; Y_2 \cond U_{12},U_{21},Q), \\
		R_1+R_2 &\le  I(U_{12},U_{21},U_{22}; Y_2\cond Q) 
			+ I(U_{11}; Y_1 \cond U_{12},U_{21},Q), \\
		R_1+R_2 &\le I(U_{11},U_{21}; Y_1 \cond U_{12},Q) 
			+ I(U_{12},U_{22}; Y_2\cond U_{21},Q), \\
		2R_1 + R_2 &\le I(U_{11},U_{12},U_{21}; Y_1 \cond Q) 
			+ I(U_{11};Y_1 \cond U_{12},U_{21},Q) 
			+ I(U_{12},U_{22}; Y_2 \cond U_{21},Q), \\
		R_1+2R_2 &\le I(U_{12},U_{21},U_{22}; Y_2\cond Q) 
			+I(U_{22}; Y_2\cond U_{12},U_{21},Q) 
			+  I(U_{11}, U_{21}; Y_1 \cond U_{12},Q)
	\end{align}
	\label{eq:R_HK_compact}%
	\end{subequations}}%
	for some pmf of the form $p = p(q)\inpmfspace p(u_{11}|q) \inpmfspace p(u_{12}|q) \inpmfspace p(u_{21}|q) \inpmfspace p(u_{22}|q) \inpmfspace p(x_1| u_{11}, u_{12}) \inpmfspace p(x_2 | u_{21}, u_{22})$, where the latter two factors represent deterministic mappings $x_1(u_{11},u_{12})$ and $x_2(u_{21},u_{22})$.

	It is easy to see that $\Rr_\text{HK} \subseteq \Rr_\text{opt}$. Choosing $\Sc_1=\{11,12,21\}$ in~\eqref{eq:R_opt1}, the resulting conditions coincide with the ones in~\eqref{eq:R_HK1}, and the constituent sets satisfy the condition $\Rr_{\text{HK},1}(p) \subseteq \Rr_1(p)$. Likewise, choosing $\Sc_2=\{12,21,22\}$ in~\eqref{eq:R_opt2}, $\Rr_{\text{HK},2}(p) \subseteq \Rr_2(p)$, and the desired inclusion follows.
	
	To show that $\Rr_\text{opt} \subseteq \Rr_\text{HK}$, note that conditions~\eqref{eq:R_opt1} and~\eqref{eq:R_opt2} must hold for some $\Sc_1 \supseteq \{11,12\}$ and $\Sc_2 \supseteq \{21,22\}$. For each of the 16 possible choices of $\Sc_1$ and $\Sc_2$, the resulting rate region is (directly or indirectly) included in $\Rr_\text{HK}$ as follows (see Figure~\ref{fig:HKcases}).
	\begin{figure}[h]
		\centering
		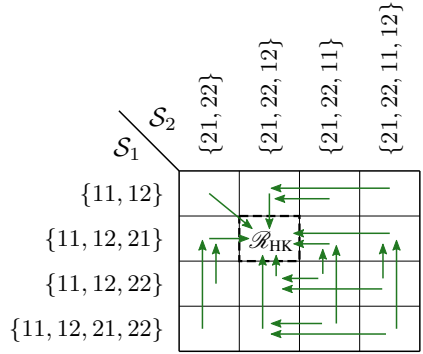
		\caption{Different cases of $\Sc_1$ and $\Sc_2$ for the region $\Rr_\text{opt}$ and the inclusion of the corresponding regions in $\Rr_\text{HK}$. An arrow from A to B means that the region achieved by case A is included in the region achieved by case B.}
		\label{fig:HKcases}
	\end{figure}
	\begin{itemize}
		\item If $\Sc_1=\{11,12,21\}$ and $\Sc_2=\{21,22,12\}$, we obtain precisely $\Rr_\text{HK}$ (depicted as a dashed box in the figure).
		
		\item If $\Sc_1 = \{11,12,21,22\}$, both receivers decode for the messages with indices $\{21,22\}$. This is equivalent to letting $U'_{21} = (U_{21},U_{22})$, $U'_{22}=\emptyset$, and  $\Sc_1'=\{11,12,21\}$. A symmetric argument holds if $\Sc_2 = \{21,22,11,12\}$.
		
		\item If $\Sc_1 = \{11,12,22\}$, then $\Sc_1$ can be replaced by $\{11,12,21\}$ by exchanging the roles of $U_{21}$ and $U_{22}$. The exchange will not affect receiver~2, since the two auxiliary random variables play symmetric roles there. A symmetric argument holds if $\Sc_2 =\{21,22,11\}$.
		
		\item If $\Sc_1=\{11,12\}$ and $\Sc_2=\{21,22\}$, we apply Fourier--Motzkin elimination and arrive at 
		\begin{align*}
			R_1 &\le I(X_1; Y_1 \cond Q), \\
			R_2 &\le I(X_2; Y_2 \cond Q).
		\end{align*}
		This region is a subset of the one in~\eqref{eq:R_HK_compact} when the latter is specialized to $U_{12}=U_{21}=\emptyset$, $U_{11}=X_1$, and $U_{22}=X_2$.
		
		\item If $\Sc_1=\{11,12\}$ and $\Sc_2=\{21,22,12\}$, Fourier--Motzkin elimination leads to 
		\begin{align*}
			R_1 &\le I(X_1; Y_1\cond Q), \\
			R_1 &\le I(X_1; Y_1\cond U_{12}, Q) + I(U_{12}; Y_2 \cond X_2, Q), \\
			R_2 &\le I(X_2; Y_2 \cond U_{12}, Q), \\
			R_1+R_2 &\le I(X_1; Y_1 \cond U_{12}, Q) + I(U_{12}, X_2; Y_2 \cond Q).
		\end{align*}
		Again, this region is a subset of the one in~\eqref{eq:R_HK_compact}, namely when the latter is specialized to $U_{21}=\emptyset$ and $U_{22}=X_2$. A symmetric argument holds if $\Sc_1=\{11,12,21\}$ and $\Sc_2=\{21,22\}$.
	\end{itemize}%
	This concludes the proof of Corollary~\ref{thm:HK_OptimalReceiver}.
\end{proof}

\section{Concluding remarks} \label{sec:conclusion}
Taking a modular approach to the problem of finding the capacity region of the interference network, we have studied the performance of random code ensembles. This result provides a simple characterization of the rate region achievable by the optimal maximum likelihood decoding rule and invites more refined studies on the performance of random coding for interference networks, such as the error exponent analysis (cf.\@ \cite{Haroutunian1975, Nazari2010}) and Verd\'u's finite-block performance bounds~\cite{Verdu2012}.

The optimal rate region can be achieved by simultaneous nonunique decoding, which can be useful in other coding schemes such as Marton coding for broadcast channels~\cite{Marton1979} and noisy network coding for relay networks~\cite{Lim2011}. Although its performance can be achieved  also by an appropriate combination of simultaneous decoding (SD) of strong interference and treating weak interference as noise (IAN) \cite{Motahari2011, Baccelli2011, Bidokhti2012}, simultaneous nonunique decoding provides a conceptual unification of SD and IAN, recovering all possible combinations of the two schemes at each receiver. Indeed, as with ``the one ring to rule them all'' \cite{Tolkien1954}, simultaneous nonunique decoding is the one rule that includes them all.

\section*{Acknowledgments}
The authors are grateful to Gerhard Kramer for an interesting conversation on optimal decoding rules that spurred interest in this research direction. They also would like to thank Jungwon Lee for enlightening discussions on minimum distance decoding
for interference channels, which shaped the main ideas behind this paper.

\appendices
\section{Proof of Lemma~\ref{thm:2MAC_singleLetterizeEntropy}}
\label{sec:proof_2MAC_singleLetterizeEntropy}
	Clearly, the right hand side of the equality is an upper bound to the left hand side, since
	\begin{align*}
		H(Y_1^n \cond X_1^n, \Cc_n) &\le nH(Y_1 \cond X_1, Q),
	\end{align*}
	and
	\begin{align*}
		H(Y_1^n \cond X_1^n, \Cc_n) &\le H(Y_1^n, M_2 \cond X_1^n, \Cc_n) \\
		&= nR_2 + H(Y_1^n \cond X_1^n, X_2^n, \Cc_n) \\
		&\le nR_2 + nH(Y_1 \cond X_1,X_2,Q),
	\end{align*}
	where we have used the codebook structure and the fact that the channel is memoryless.

	To see that the right hand side is also a valid lower bound, note that
	\begin{align*}
		H(Y_1^n\cond X_1^n,\Cc_n) 
		&= \underbrace{H(Y_1^n\cond X_1^n, \Cc_n, M_2)}_{\substack{= nH(Y_1|X_1,X_2) \\ = nH(Y_1| X_1,X_2,Q)}} + \underbrace{H(M_2)}_{=nR_2} - H(M_2\cond X_1^n, \Cc_n, Y_1^n).
	\end{align*}
	Next, we find an upper bound on $H(M_2\cond X_1^n, \Cc_n, Y_1^n)$ by showing that given $X_1^n$, $\Cc_n$, and $Y_1^n$, a relatively short list $\Lc \subseteq \natSet{2^{nR_2}}$ can be constructed that contains $M_2$ with high probability (the idea is similar to the proof of Lemma~22.1 in~\cite{ElGamalKim}). Without loss of generality, assume $M_2=1$.
	Fix an $\eps>0$ and define the random set
	\begin{align*}
		\Lc &= \{m_2: (Q^n, X_1^n, X_2^n(m_2), Y_1^n) \in \aep \}.
	\end{align*}
	To analyze the cardinality $\card{\Lc}$, note that, for each $m_2\neq 1$,
	{\allowdisplaybreaks 
	\begin{align*}
		\P\{ (Q^n, X_1^n, X_2^n(m_2), Y_1^n) \in \aep \}  
		&= \sum_{q^n,x_1^n,x_2^n} \P\{Q^n=q^n, X_1^n=x_1^n, X_2^n(m_2)=x_2^n\} 
		\, \P\{(x_1^n, x_2^n, Y_1^n) \in \aep\} \\
		&\annleq{a} \sum_{q^n,x_1^n,x_2^n} \P\{Q^n=q^n, X_1^n=x_1^n, X_2^n(m_2)=x_2^n\} 
		\ 2^{-n(I(X_2;Y_1|X_1,Q)-\delta(\eps))} \\[1mm]
		&= 2^{-n(I(X_2;Y_1|X_1,Q)-\delta(\eps))},
	\end{align*}}%
	where (a) follows by the joint typicality lemma. Thus, the cardinality $\card{\Lc}$ satisfies $\card{\Lc} \le 1 + B$, where $B$ is a binomial random variable with $2^{nR_2}-1$ trials and success probability at most $2^{-n(I(X_2;Y_1| X_1,Q)-\delta(\eps))}$. The expected cardinality is therefore bounded as
	\begin{align}
		\E(\card{\Lc}) &\le 1 + 2^{n(R_2-I(X_2;Y_1| X_1,Q)+\delta(\eps))}.
		\label{eq:listCardBound}
	\end{align}
	Note that the true $M_2$ is contained in the list with high probability, i.e., $1\in\Lc$, by the weak law of large numbers, 
	\begin{align*}
		& \P\{ (Q^n,X_1^n, X_2^n(1), Y_1^n) \in \aep \} \to 1\quad \text{as }n\to\infty.
	\end{align*}
	Define the indicator random variable $E = \mathbb I(1 \in \Lc)$, which therefore satisfies $\P\{E=0\} \to 0 $ as $n\to\infty$. Hence
	\begin{align*}
		H(M_2\cond X_1^n,\Cc_n,Y_1^n) 
		&= H(M_2\cond X_1^n,\Cc_n,Y_1^n,E) + I(M_2;E\cond X_1^n,\Cc_n,Y_1^n) \\
		&\le H(M_2\cond X_1^n,\Cc_n,Y_1^n,E) + 1 \\
		&= 1+ \P\{E=0\}\cdot H(M_2\cond X_1^n,\Cc_n,Y_1^n,E=0) \\
		&\qquad + \P\{E=1\}\cdot H(M_2\cond X_1^n,\Cc_n,Y_1^n,E=1) \\
		&\le 1+ nR_2 \P\{E=0\} + H(M_2\cond X_1^n,\Cc_n,Y_1^n,E=1).
	\end{align*}
	For the last term, we argue that if $M_2$ is included in $\Lc$, then its conditional entropy cannot exceed $\log(\card{\Lc})$:
	{\allowdisplaybreaks
	\begin{align*}
		H(M_2\cond X_1^n,\Cc_n,Y_1^n,E=1) 
		&\anneq{a} H(M_2\cond X_1^n,\Cc_n,Y_1^n,E=1,\Lc,\card{\Lc}) \\
		&\le H(M_2\cond E=1,\Lc,\card{\Lc}) \\
		&= \sum_{l=0}^{2^{nR_2}} \P\{\card{\Lc}=l\} \cdot H(M_2\cond E=1,\Lc,\card{\Lc}=l) \\
		&\le \sum_{l=0}^{2^{nR_2}} \P\{\card{\Lc}=l\} \cdot \log(l) \\
		&= \E (\log(\card{\Lc})) \\
		&\annleq{b} \log( \E (\card{\Lc}) ) \\
		&\annleq{c} 1 + \max\{ 0, n(R_2-I(X_2;Y_1| X_1,Q)+\delta(\eps)) \},
	\end{align*}}%
	where (a) follows since the list $\Lc$ and its cardinality $\card{\Lc}$ are functions only of $X_1^n$, $\Cc_n$, and $Y_1^n$, (b) follows by Jensen's inequality, and (c) follows from~\eqref{eq:listCardBound} and the soft-max interpretation of the log-sum-exp function~\cite[p.\,72]{BoydVandenberghe}.
	
	Substituting back, we have
	\begin{align*}
		H(M_2\cond X_1^n, \Cc_n, Y_1^n) 
		&\le 2+ nR_2\P\{E=0\} + \max\{0, n(R_2-I(X_2;Y_1| X_1,Q)+\delta(\eps))\},
	\end{align*}
	and
	\begin{align*}
		\sfrac 1 n H(Y_1^n\cond X_1^n, \Cc_n) 
 		&
		\ge H(Y_1|X_1,X_2,Q) + R_2 -\sfrac 2 n - R_2\P\{E=0\} - \max\{0, R_2 - I(X_2;Y_1| X_1,Q) + \delta(\eps) \} \\
 		&
		\ge H(Y_1|X_1,X_2,Q) + \min\{R_2, I(X_2;Y_1| X_1,Q) - \delta(\eps) \} -\sfrac 2 n - R_2\P\{E=0\} .
	\end{align*}
	Taking the limit as $n\to\infty$, and noting that we are free to choose $\eps$ such that $\delta(\eps)$ becomes arbitrarily small, the desired result follows.

\section{Equivalence between the Min and MAC Forms}
\label{sec:KLdmic_minMacFormEquivalence}
Fix a distribution $p=p(q)\inpmfspace p(x_1|q)\cdots p(x_K|q)$ and a rate tuple $(R_1,\dots,R_K)$. We show that the conditions~\eqref{eq:KLdmic_R1macformCond} and~\eqref{eq:KLdmic_R1minformCond} are equivalent.

\begin{proof}[Proof that (\ref{eq:KLdmic_R1macformCond}) implies (\ref{eq:KLdmic_R1minformCond})]
We are given a set $\Sc$ with $\Dc_1 \subseteq \Sc \subseteq \natSet K$. Fix an arbitrary $\Vc$ with nonempty intersection $\Vc \cap \Dc_1$. Now consider $\Vc' = \Tc = \Sc \cap \Vc$. Note $\Vc'\cap\Dc_1 = \Vc\cap\Dc_1$ as required. Then, 
\begin{align*}
	R_{\Vc'} = R_{\Tc} 
	&\annleq{a} I(X_{\Tc}; Y_1 \cond X_{\Sc\setminus\Tc}, Q) \\
	&\annleq{b} I(X_{\Tc}; Y_1 \cond X_{\Sc\setminus\Vc}, X_{\natSet K \setminus \Sc\setminus\Vc}, Q) \\
	&= I(X_{\Vc'}; Y_1 \cond X_{\natSet K \setminus \Vc}, Q),
\end{align*}
where (a) follows from~\eqref{eq:KLdmic_R1macformCond}, and (b) follows from the structure of $p$.
\end{proof}

\begin{proof}[Proof that (\ref{eq:KLdmic_R1minformCond}) implies (\ref{eq:KLdmic_R1macformCond})]
Denote a set $\Sc \subseteq \natSet K$ as \emph{decodable} if 
\begin{align*}
	\forall \Tc \subseteq \Sc: \quad R_{\Tc} &\leq I(X_{\Tc}; Y_1 \cond X_{\Sc\setminus\Tc}, Q).
\end{align*}
Then the following proposition holds, which is proved below.
\begin{proposition} \label{thm:KLdmic_union}  
	If $\Sc_1$ and $\Sc_2$ are decodable sets, then $\Sc_1 \cup \Sc_2$ is a decodable set.
\end{proposition}

To determine which messages are decodable, consider the optimization problem of maximizing $\card{\Sc}$ over decodable sets $\Sc$. From Proposition~\ref{thm:KLdmic_union}, a unique maximizer $\Scs$ must exist, which is a superset of all decodable sets. Consider its complement $\Scn$. The intuitive reason for the messages indexed by $\Scn$ being \emph{un}decodable is that the corresponding rates are too large. This notion is made precise in the following proposition, which is analogous to a property for the Gaussian case given in~\cite[Fact 1]{Baccelli2011} and for which a proof is provided below.
\begin{proposition} \label{thm:KLdmic_undecodableRates}
	For all sets $\Uc$ with $\emptyset \subset \Uc \subseteq \Scn$, the rates satisfy
	\begin{align}
		R_\Uc &> I(X_\Uc; Y_1 \cond X_{\Scs},Q).
		\label{eq:KLdmic_buildContradiction}
	\end{align}
\end{proposition}

Assuming~\eqref{eq:KLdmic_R1macformCond} is not true, there must be some desired message index that is not decodable, i.e., $\Dc_1 \nsubseteq \Scs$, or equivalently, $ \Scn \cap \Dc_1 \neq \emptyset$. Then we can choose $\Vc = \Scn$ in~\eqref{eq:KLdmic_R1minformCond}, yielding
\begin{align*}
	\exists \Vc' \subseteq \Scn, \Vc'\cap\Dc_1 = \Scn\cap\Dc_1: \quad 
	R_{\Vc'} & \leq I(X_{\Vc'}; Y_1 \cond X_{\Scs},Q), 
\end{align*}
which contradicts~\eqref{eq:KLdmic_buildContradiction}. This proves that (\ref{eq:KLdmic_R1minformCond}) implies (\ref{eq:KLdmic_R1macformCond}).
\end{proof}

\smallskip

\begin{proof}[Proof of Proposition~\ref{thm:KLdmic_union}]
	Since $\Sc_1$ and $\Sc_2$ are decodable, we have 
	\begin{align*}
		R_{\Tc} &\leq I(X_{\Tc}; Y_1 \cond X_{\Sc_1 \setminus\Tc}, Q) \quad \text{ for all $\Tc \subseteq \Sc_1$}, \\
		R_{\Tc'} &\leq I(X_{\Tc'}; Y_1 \cond X_{\Sc_2 \setminus\Tc'}, Q) \quad \text{ for all $\Tc' \subseteq \Sc_2$}.
	\end{align*}
	and we need to show
	\begin{align*}
		R_{\Tc''} &\leq I(X_{\Tc''}; Y_1 \cond X_{(\Sc_1 \cup \Sc_2) \setminus\Tc''}, Q) \quad \text{ for all $\Tc'' \subseteq \Sc_1 \cup \Sc_2$}.
	\end{align*}
	Fix a subset $\Tc'' \subseteq \Sc_1 \cup \Sc_2$ and partition it as $\Tc'' = \Tc''_1 \cup \Tc''_2$ where $\Tc''_1 \subseteq \Sc_1$, $\Tc''_2 \subseteq \Sc_2$, $\Tc''_1 \cap \Tc''_2 = \emptyset$, and $\Tc''_2 \cap \Sc_1 = \emptyset$ (see Figure~\ref{fig:unionIsDecodable}). 

	\begin{figure}[h!]
		\centering
		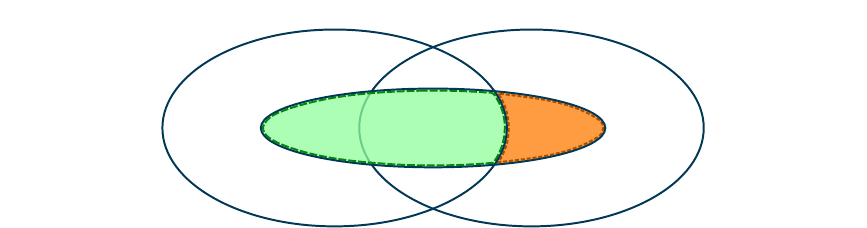
		\caption{Partitioning the set $\Tc'' \subseteq \Sc_1 \cup \Sc_2$.}
		\label{fig:unionIsDecodable}
	\end{figure}

	Then
	\begin{align*}
		R_{\Tc''} &= R_{\Tc''_1} + R_{\Tc''_2} \\
		&\leq I(X_{\Tc''_1}; Y_1 \cond X_{\Sc_1 \setminus\Tc''_1}, Q) + I(X_{\Tc''_2}; Y_1 \cond X_{\Sc_2 \setminus\Tc''_2}, Q) \\
		&\leq I(X_{\Tc''_1}; Y_1 \cond X_{(\Sc_1 \cup \Sc_2)\setminus\Tc''}, Q) + I(X_{\Tc''_2}; Y_1 \cond X_{(\Sc_1 \cup \Sc_2)\setminus\Tc''}, X_{\Tc''_1},Q) \\
		& = I( X_{\Tc''_1}, X_{\Tc''_2}; Y_1 \cond X_{(\Sc_1 \cup \Sc_2)\setminus\Tc''}, Q),
	\end{align*}
	which concludes the proof.
\end{proof}

\smallskip

\begin{proof}[Proof of Proposition~\ref{thm:KLdmic_undecodableRates}]
	Assume first that the proposition was not true. Then there must be a minimal $\Uc$ with $\emptyset \subset \Uc \subseteq \Scn$ such that 
	\begin{align}
		R_\Uc &\leq I(X_\Uc; Y_1 \cond X_{\Scs},Q), \label{eq:KLdmic_RUviolates} \\
		R_{\Uc\setminus \Tc} &> I(X_{\Uc\setminus \Tc}; Y_1 \cond X_{\Scs},Q) \quad \text{ for all $\Tc$ with $\emptyset \subset \Tc \subset \Uc$}.  \notag
	\end{align}
	Now, 
	\begin{align*}
		R_{\Tc} = R_\Uc - R_{\Uc \setminus \Tc} 
		&\leq I(X_\Uc; Y_1 \cond X_{\Scs},Q) - I(X_{\Uc \setminus \Tc}; Y_1 \cond X_{\Scs},Q) \notag\\
		&= I(X_\Tc; Y_1 \cond X_{\Scs}, X_{\Uc \setminus \Tc}, Q ) \quad \text{ for all $\Tc$ satisfying $\emptyset \subset \Tc \subset \Uc$}.
	\end{align*}
	Recalling~\eqref{eq:KLdmic_RUviolates}, the last statement continues to hold for $\Tc=\Uc$. Thus,
	\begin{align}
		R_\Tc  
		&\leq I(X_\Tc; Y_1 \cond X_{\Scs}, X_{\Uc \setminus \Tc}, Q ) \quad \text{ for all $\Tc \subseteq \Uc$} .
		\label{eq:KLdmic_RUminusRUprime}
	\end{align}
	We are going to show that $\Scs \cup \Uc$ is decodable, which contradicts the definition of $\Scs$ as the maximum decodable set since $\Uc$ is non-empty and does not intersect $\Scs$. To this end, consider an arbitrary $\Tc' \subseteq \Scs \cup \Uc$ and partition it as $\Tc' = \Tc'_1 \cup \Tc'_2$ with $\Tc'_1 \cap \Tc'_2 = \emptyset$, $\Tc'_1 \subseteq \Scs$, and $\Tc'_2 \subseteq \Uc$ (see Figure~\ref{fig:ratesAreTooLarge}). 
	
	\begin{figure}[h!]
		\centering
		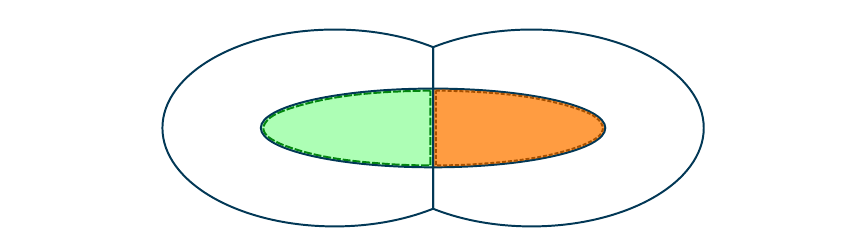
		\caption{Partitioning the set $\Tc'  \subseteq \Scs \cup \Uc$.}
		\label{fig:ratesAreTooLarge}
	\end{figure}

	Then
	\begin{align*}
		R_{\Tc'} &= R_{\Tc'_1} + R_{\Tc'_2} \\
		&\annleq{a} I(X_{\Tc'_1}; Y_1 \cond X_{\Scs \setminus\Tc'_1}, Q) + I(X_{\Tc'_2}; Y_1 \cond X_{\Scs}, X_{\Uc \setminus\Tc'_2}, Q) \\
		&\annleq{b} I(X_{\Tc'_1}; Y_1 \cond X_{\Scs \setminus\Tc'_1}, X_{\Uc \setminus \Tc'_2},Q) + I(X_{\Tc'_2}; Y_1 \cond X_{\Scs \setminus\Tc'_1}, X_{\Uc \setminus \Tc'_2}, X_{\Tc'_1}, Q) \\
		& = I( X_{\Tc'_1}, X_{\Tc'_2}; Y_1 \cond X_{(\Scs \cup \Uc)\setminus(\Tc'_1 \cup \Tc'_2)}, Q),
	\end{align*}
	where (a) follows from $\Scs$ being decodable and~\eqref{eq:KLdmic_RUminusRUprime}, and in (b), we have augmented the first mutual information expression and rewritten the second one. This concludes the proof by contradiction.
\end{proof}

\section{Proof of Lemma~\ref{thm:KLMAC_singleLetterizeEntropy}}
\label{sec:proof_KLMAC_singleLetterizeEntropy}
	The proof proceeds along similar lines as the proof of Lemma~\ref{thm:2MAC_singleLetterizeEntropy}. 
	First, we show that the right hand side is a valid upper bound to the left hand side. For any $\Uc \subseteq \compl{\Sc}$, 
	\begin{align*}
		H(Y_1^n \cond X_{\Sc}^n, \Cc_n) &\le H(Y_1^n, M_{\Uc} \cond X_{\Sc}^n, \Cc_n) \\
		&= nR_{\Uc} + H(Y_1^n \cond X_{\Sc}^n, X_{\Uc}^n, \Cc_n) \\
		&\le nR_{\Uc} + nH(Y_1 \cond X_{\Sc},X_{\Uc},Q) \\
		&= nR_{\Uc} + nH(Y_1 \cond X_{\natSet K},Q) + I( X_{\compl{(\Sc\cup\Uc)}}; Y_1 \cond X_{\Sc\cup\Uc}, Q),
	\end{align*}
	where we have used the codebook structure.
	
	To see that the right hand side is a valid lower bound to the left hand side, note
	\begin{align*}
		H(Y_1^n\cond X_{\Sc}^n, \Cc_n) 
		&= nH(Y_1| X_{\natSet K},Q) + nR_{\compl{\Sc}} - H(M_{\compl{\Sc}} \cond X_{\Sc}^n, Y_1^n, \Cc_n).
	\end{align*}
	Without loss of generality, assume $M_k=1$, for $k\in \compl{\Sc}$. 
	Fix an $\eps>0$ and define the random set
	\begin{align*}
		\Lc &= \{m_{\compl{\Sc}}: (Q^n, X_i^n\vert_{i\in \Dc_1}, X_i^n(m_i)\vert_{i \in \compl{\Dc_1}}, Y_1^n) \in \aep \text{ with $m_k=1$ for all $k\in\compl{\Dc_1}\cap\Sc$}\}.
	\end{align*}
	To analyze the cardinality $\card{\Lc}$, fix a $m_{\compl{\Sc}}$ and consider the probability of $m_{\compl{\Sc}} \in \Lc$. 
	If $m_k\neq 1$ for all $k\in \compl{\Sc}$, and $m_k=1$ otherwise, then the joint typicality lemma implies
	\begin{align*}
		\P\{ (Q^n, X_i^n\vert_{i\in \Dc_1}, X_i^n(m_i)\vert_{i \in \compl{\Dc_1}}, Y_1^n) \in \aep \} 
		&\le 2^{-n(I(X_{\compl{\Sc}};Y_1|X_{\Sc},Q)-\delta(\eps))},
	\end{align*}
	and there are at most $2^{nR_{\compl{\Sc}}}$ such $m_{\compl{\Sc}}$. 
	More generally, fix a subset $\Uc \subseteq \compl{\Sc}$. If $m_k \neq 1$ for $k\in \compl{\Sc}\setminus\Uc$, and $m_k = 1$ otherwise, then
	\begin{align*}
		\P\{ (Q^n, X_i^n\vert_{i\in \Dc_1}, X_i^n(m_i)\vert_{i \in \compl{\Dc_1}}, Y_1^n) \in \aep \}  
		&\le 2^{-n(I(X_{\compl{\Sc}\setminus\Uc};Y_1|X_{\Sc},X_{\Uc},Q)-\delta(\eps))},
	\end{align*}
	and there are at most $2^{nR_{\compl{\Sc}\setminus\Uc}}$ such $m_{\compl{\Sc}}$.
	Thus,
	\begin{align}
		\E(\card{\Lc}) &\le \sum_{\Uc \subseteq \compl{\Sc}} 2^{n(R_{\compl{\Sc}\setminus\Uc}-I(X_{\compl{\Sc}\setminus\Uc};Y_1|X_{\Sc},X_{\Uc},Q)+\delta(\eps))}.
		\label{eq:KL_prop2_listCardBound}
	\end{align}
	Define the indicator random variable $E = \mathbb I((1,1,\dots,1) \in \Lc)$, which satisfies $\P\{E=0\} \to 0 $ as $n\to\infty$ by the weak law of large numbers. Now
	\begin{align*}
		H(M_{\compl{\Sc}} \cond X_{\Sc}^n, Y_1^n, \Cc_n)
		&\le 1+ nR_{\compl{\Sc}} \P\{E=0\} + H(M_{\compl{\Sc}} \cond X_{\Sc}^n, Y_1^n, \Cc_n,E=1).
	\end{align*}
	For the last term, we argue 
	\begin{align*}
		H(M_{\compl{\Sc}} \cond X_{\Sc}^n, Y_1^n, \Cc_n,E=1) 
		&\le \log( \E (\card{\Lc}) ) \\
		&\annleq{\ref{eq:KL_prop2_listCardBound}} \log\left(\sum_{\Uc \subseteq \compl{\Sc}} 2^{n(R_{\compl{\Sc}\setminus\Uc}-I(X_{\compl{\Sc}\setminus\Uc};Y_1|X_{\Sc},X_{\Uc},Q)+\delta(\eps))} \right) \\
		&\le \max_{\Uc \subseteq \compl{\Sc}} \bigl( n(R_{\compl{\Sc}\setminus\Uc}-I(X_{\compl{\Sc}\setminus\Uc};Y_1|X_{\Sc},X_{\Uc},Q)+\delta(\eps)) \bigr) + \card{\compl{\Sc}}. 
	\end{align*}
	Substituting back, 
	\begin{align*}
		H(M_{\compl{\Sc}} \cond X_{\Sc}^n, Y_1^n, \Cc_n) 
		&\le 1 + \card{\compl{\Sc}} + nR_{\compl{\Sc}} \P\{E=0\} + \max_{\Uc \subseteq \compl{\Sc}} \bigl( n(R_{\compl{\Sc}\setminus\Uc}-I(X_{\compl{\Sc}\setminus\Uc};Y_1|X_{\Sc},X_{\Uc},Q)+\delta(\eps)) \bigr) ,
	\end{align*}
	and finally,
	\begin{align*}
		\sfrac 1 n H(Y_1^n\cond X_{\Sc}^n, \Cc_n) 
		&\ge H(Y_1| X_{\natSet K},Q) + R_{\compl{\Sc}} - \sfrac {1 + \card{\compl{\Sc}}}{n} - R_{\compl{\Sc}} \P\{E=0\} \\
		&\qquad - \max_{\Uc \subseteq \compl{\Sc}} \bigl( R_{\compl{\Sc}\setminus\Uc}-I(X_{\compl{\Sc}\setminus\Uc};Y_1|X_{\Sc},X_{\Uc},Q)+\delta(\eps) \bigr) \\
		&= H(Y_1| X_{\natSet K},Q) - \sfrac {1 + \card{\compl{\Sc}}}{n} - R_{\compl{\Sc}} \P\{E=0\} \\
		&\qquad + \min_{\Uc \subseteq \compl{\Sc}} \bigl( R_{\Uc}+ I(X_{\compl{\Sc}\setminus\Uc};Y_1|X_{\Sc},X_{\Uc},Q)+\delta(\eps) \bigr).
	\end{align*}
	Taking the limits $n\to\infty$ and $\eps \to 0$ concludes the proof.

\bibliographystyle{IEEEtranS}
\bibliography{IEEEabrv,references-unified,myPublications}

\end{document}